\documentclass{article}

\newcommand{\aop}{Ann. Phys.~}
\newcommand{\cmp}{Comm. Math. Phys.~}

\newcommand{\jmp}{J. Math. Phys.~}
\newcommand{\jpa}{J. Phys. A~}

\newcommand{\prl}{Phys. Rev. Lett.~}

\newcommand{\pla}{Phys. Lett. A~}

\newcommand{\laa}{Lin. Alg. App.~}

\usepackage{appendix}
\usepackage{epstopdf}
\usepackage{color}
\usepackage{subfigure}
\usepackage[T1]{fontenc}
\usepackage[sc]{mathpazo}
\usepackage{amsmath}
\usepackage{amssymb}
\usepackage{graphicx,color}
\usepackage{framed}
\usepackage{multirow}
\usepackage{enumerate}
\usepackage{amsthm}
\usepackage{amsfonts,mathrsfs}
\usepackage{geometry} 
\usepackage{eepic}
\usepackage{ifthen}
\usepackage[vcentermath]{youngtab}
\usepackage[unicode=true,pdfusetitle, bookmarks=true,bookmarksnumbered=false,bookmarksopen=false, breaklinks=false,pdfborder={0 0 0},backref=false,colorlinks=false] {hyperref}
\hypersetup{
colorlinks,linkcolor=myurlcolor,citecolor=myurlcolor,urlcolor=myurlcolor}
\definecolor{myurlcolor}{rgb}{0,0,0.7}

\usepackage{color}

\newcommand{\blue}{\textcolor{blue}}

\usepackage{hyperref}
\hypersetup{pdfpagemode=UseNone}

\newcommand{\tinyspace}{\mspace{1mu}}

\newcommand{\op}[1]{\operatorname{#1}}

\newcommand{\abs}[1]{\left\lvert\tinyspace #1 \tinyspace\right\rvert}

\newcommand{\norm}[1]{\left\lVert\tinyspace #1 \tinyspace\right\rVert}

\renewcommand{\t}{{\scriptscriptstyle\mathsf{T}}}

\newcommand{\setft}[1]{\mathrm{#1}}

\newcommand{\density}[1]{\setft{D}\left(#1\right)}

\newcommand{\sign}{\op{sign}}

\def\g{\mathfrak{g}}

\def\k{\mathfrak{k}}

\def\liet{\mathfrak{t}}
\def\su{\mathfrak{su}}

\def\Ad{\mathrm{Ad}}
\def\ad{\mathrm{ad}}
\def\vol{\mathrm{vol}}
\def\dh{\mathrm{DH}}

\def \dif {\mathrm{d}}
\def \diag {\mathrm{diag}}
\def \vol {\mathrm{vol}}

\def\complex{\mathbb{C}}
\def\real{\mathbb{R}}
\def\natural{\mathbb{N}}

\def\I{\mathbb{1}}

\newenvironment{mylist}[1]{\begin{list}{}{
    \setlength{\leftmargin}{#1}
    \setlength{\rightmargin}{0mm}
    \setlength{\labelsep}{2mm}
    \setlength{\labelwidth}{8mm}
    \setlength{\itemsep}{0mm}}}
    {\end{list}}

\def\ot{\otimes}

\newcommand{\inner}[2]{\langle #1 , #2\rangle}

\newcommand{\Inner}[2]{\left\langle #1 , #2\right\rangle}

\newcommand{\defeq}{\stackrel{\smash{\textnormal{\tiny def}}}{=}}

\newcommand{\Pa}[1]{\left(#1\right)}

\newcommand{\Br}[1]{\left[#1\right]}
\newcommand{\set}[1]{\{#1\}}
\newcommand{\Set}[1]{\left\{#1\right\}}

\newcommand{\ket}[1]{|#1\rangle}

\DeclareMathOperator{\trace}{Tr}

\newcommand{\Ptr}[2]{\trace_{#1}\Pa{#2}}

\newcommand{\Tr}[1]{\Ptr{}{#1}}

\newcommand{\Abs}[1]{\left|\tinyspace#1\tinyspace\right|}

\def\cB{\mathcal{B}}\def\cE{\mathcal{E}}
\def\cF{\mathcal{F}}\def\cG{\mathcal{G}}
\def\cO{\mathcal{O}}

\def\cX{\mathcal{X}}\def\cY{\mathcal{Y}}

\def\bP{\mathbf{P}}

\def\bsA{\boldsymbol{A}}\def\bsB{\boldsymbol{B}}

\def\bsU{\boldsymbol{U}}

\def\bsa{\boldsymbol{a}}\def\bsb{\boldsymbol{b}}
\def\bsh{\boldsymbol{h}}

\def\bsp{\boldsymbol{p}}\def\bsq{\boldsymbol{q}}\def\bsr{\boldsymbol{r}}\def\bss{\boldsymbol{s}}
\def\bsu{\boldsymbol{u}}\def\bsv{\boldsymbol{v}}\def\bsx{\boldsymbol{x}}

\def\rG{\mathrm{G}}\def\rH{\mathrm{H}}
\def\rL{\mathrm{L}}
\def\rS{\mathrm{S}}
\def\rU{\mathrm{U}}

\def\sC{\mathscr{C}}\def\sD{\mathscr{D}}

\def\sS{\mathscr{S}}

\def\T{\textsf{T}}

\newtheorem{thrm}{Theorem}[section]

\newtheorem{prop}[thrm]{Proposition}
\newtheorem{cor}[thrm]{Corollary}
\theoremstyle{definition}
\newtheorem{definition}[thrm]{Definition}
\newtheorem{remark}[thrm]{Remark}
\newtheorem{exam}[thrm]{Example}

\numberwithin{equation}{section}

\newcounter{questionnumber}

\begin{document}

\title{\Large \bf Duistermaat-Heckman measure and the mixture of quantum states}

\author{\blue{Lin Zhang}$^{1,2}$\footnote{E-mail: godyalin@163.com; linzhang@mis.mpg.de},\, \blue{Yixin Jiang}$^1$ and \blue{Junde Wu}$^3$\\
  {\it\small $^1$Institute of Mathematics, Hangzhou Dianzi University, Hangzhou 310018, PR~China}\\
  {\it\small $^2$Max-Planck-Institute for Mathematics in the Sciences, Leipzig 04103,
  Germany}\\
  {\it\small $^3$School of Mathematical Sciences, Zhejiang University, Hangzhou 310027, PR~China}}
\date{}
\maketitle
\begin{abstract}

In this paper, we present a general framework to solve a fundamental
problem in Random Matrix Theory (RMT), i.e., the problem of
describing the joint distribution of eigenvalues of the sum
$\bsA+\bsB$ of two independent random Hermitian matrices $\bsA$ and
$\bsB$. Some considerations about the mixture of quantum states are
basically subsumed into the above mathematical problem. Instead, we
focus on deriving the spectral density of the mixture of adjoint
orbits of quantum states in terms of Duistermaat-Heckman measure,
originated from the theory of symplectic geometry. Based on this
method, we can obtain the spectral density of the mixture of
independent random states. In particular, we obtain explicit
formulas for the mixture of random qubits. We also find that, in the
two-level quantum system, the average entropy of the equiprobable
mixture of $n$ random density matrices chosen from a random state
ensemble (specified in the text) increases with the number $n$.
Hence, as a physical application, our results quantitatively explain
that the quantum coherence of the mixture monotonously decreases
statistically as the number of components $n$ in the mixture.
Besides, our method may be used to investigate some statistical
properties of a special subclass of unital qubit channels.

\end{abstract}

\section{Introduction}

According to one of postulates in Quantum Mechanics, the pure state
of a single quantum system is represented by a vector in a complex
Hilbert space. It is also well-known that Hilbert space of a
composite quantum system is characterized by the tensor product of
the Hilbert spaces of the individual components. Clearly the
dimension of a composite quantum system grows exponentially with the
number of their components. This leads to exponential complexity.
Recently, Christandl \emph{et al} in \cite{Christandl2014} had
presented an effective method in order to get some physical features
which depend only on the eigenvalues of the one-body reduced states
of a randomly-chosen multipartite quantum state. Let us briefly
recall their work here. In more detail, they have described an
explicit algorithm to compute the joint eigenvalue distribution of
all the reduced density matrices of a pure multipartite quantum
state drawn at random from the unitarily invariant distribution.
Moreover, the situation for the mixed state can always be reduced to
the pure state case by the purification technique. Mathematically,
the eigenvalue distributions obtained are just Duistermaat-Heckman
measures on the moment polytope \cite{Guillemin2009}. Here the
moment polytope \cite{Burgisser2018}, i.e., the support of
Duistermaat-Heckman measure, is the solution of the one-body quantum
marginal problem, i.e., the problem of identifying the set of
possible reduced density matrices; the Duistermaat-Heckman measure
is defined to be the push-forward of the Liouville measure on a
symplectic manifold along a moment map. Later, some specific
examples (lower dimensional computation) are given based on their
algorithm solution to the one-body quantum marginal problem. In
particular, some eigenvalue distributions involved in qubits can be
explicitly illustrated. As noted by the authors, the well-known
Horn's problem \cite{Horn1962}, i.e., the determination of the
possible eigenvalues of the sum of two Hermitian matrices with fixed
eigenvalues, is a specific application of the one-body quantum
marginal problem. Moreover, they used their approach to easily
recover the main result, that is, \eqref{eq:NA-DH-measure}, obtained
in \cite{Dooley1993}. But, however, they just obtained only abstract
formula \eqref{eq:NA-DH-measure} with no explicit expressions and
without applying it to study random states in quantum information
theory. In fact, the (probabilistic) mixture (i.e., the convex
combination) of quantum states is necessarily encountered in quantum
information theory. For example, the mixture of quantum states
arises in the convexity of the entanglement measure and the
coherence measure, etc. Thus, it is necessary to figure out the
spectral density of the mixture when we use the technique of Random
Matrix Theory (RMT) \cite{Mehta2004} to study relevant problems.
Motivated by this, in the paper, we will focus on
Duistermaat-Heckman measure over the moment polytope corresponding
to the Horn's problem and provide some analytical computations in
lower dimensional spaces which are perhaps related to some problems
in quantum information theory. Another motivation about this
investigation is perhaps related to the well-known fact---the
distribution law of the sum of independent random variables is
described by the convolution of the distribution law of individual
random variables---in Probability Theory. Instead, what we will
consider in the paper is to describe the spectral law of the sum of
non-commutative random variables, i.e., random Hermitian matrices.
In particular, we focus on the distribution law of the mixture of
random quantum states. The framework introduced in the work can just
applies to such problem\footnote{Note that recently, Zuber
\cite{JBZ2017} made a similar research with focus on the
distribution of spectrum of sum of two Hermitian matrices with given
spectra instead of the mixture of random quantum states. In his
method, the key point is to do some special kind of integrals
(apparently a difficult problem when the dimension increases).
Besides, our methods used in the paper are completely different from
Zuber's.}.

To be more specific, we consider the following problem: To derive
the spectral density of the equiprobable mixture of $n$ random
density matrices (i.e., positive semi-definite complex matrix of
unit trace), each of them chosen from an adjoint orbit $\cO_{\rho}$.
Here the adjoint orbit $\cO_\rho$ of $\rho$ is the set of
isospectral density matrices, that is, $\cO_{\rho} = \Set{\bsU\rho
\bsU^\dagger: \bsU\in\rS\rU(d)}$, where $\rS\rU(d)$ is the special
unitary group of $d\times d$ unitary complex matrices. Apparently,
$\cO_{\rho}$ is essentially determined by the spectrum of $\rho$. If
we write $\boldsymbol{\lambda}=(\lambda_1,\ldots,\lambda_d)$ for the
spectrum of $\rho$, i.e., a probability vector with
$\lambda_1\geqslant\cdots\geqslant \lambda_d\geqslant0$ and
$\sum^d_{j=1}\lambda_j=1$, then
$\cO_\rho=\cO_{\boldsymbol{\lambda}}$. With the above notations, our
problem can be reformulated as: Given $n$ probability vectors
$\boldsymbol{\lambda}^1,\ldots,\boldsymbol{\lambda}^n$, we derive
the spectral density of the mixture:
\begin{eqnarray}\label{eq:multi-orbits}
\rho_s = \frac1n \Pa{\sum^n_{j=1}\bsU_j \boldsymbol{\lambda}^j
\bsU^\dagger_j},
\end{eqnarray}
where each $\bsU_j\in\rS\rU(d)$ distributed according to the
normalized Haar measure. Note here that we make abuse of notations:
all $\boldsymbol{\lambda}^j$ can be viewed as diagonal matrices
$\diag\Pa{\lambda^{(j)}_1,\ldots,\lambda^{(j)}_d}$, or vectors
$\Pa{\lambda^{(j)}_1,\ldots,\lambda^{(j)}_d}$ according to the
context. Interestingly, if all $\boldsymbol{\lambda}^j$ are equal to
the same $\boldsymbol{\lambda}$, then \eqref{eq:multi-orbits} can be
viewed as a image of a random mixed-unitary channel $\Phi$:
$\rho_s=\Phi(\boldsymbol{\lambda})$.

Once we work out the spectral density of \eqref{eq:multi-orbits}, we
can use this to derive the spectral density of the equiprobable
mixture:
\begin{eqnarray}\label{eq:equi-mix}
\rho = \frac1n\sum^n_{j=1}\rho_j,
\end{eqnarray}
where $\rho_j\in\density{\complex^d}$, the set of all $d\times d$
density matrices, for $j=1,\ldots,n$ are chosen independently from
the following random state ensemble $\cE_{d,k}(d\leqslant k)$ which
are explained immediately. Such ensemble $\cE_{d,k}$ is obtained by
partial-tracing over the $k$-dimensional subsystem of a
$dk$-dimensional composite quantum system in pure states which are
Haar-distributed. Because we have already known the eigenvalue
distribution of the ensemble $\cE_{d,k}$ \cite{Zyczkowski2001},
i.e.,
\begin{eqnarray}\label{eq:pdk}
p_{d,k}(\boldsymbol{\lambda}):=p_{d,k}(\lambda_1,\ldots,\lambda_d)=C_{d,k}\delta\Pa{1-\sum^d_{j=1}\lambda_j}\prod^d_{j=1}\lambda^{k-d}_j\theta(\lambda_j)
\prod_{1\leqslant i<j\leqslant d}(\lambda_i-\lambda_j)^2.
\end{eqnarray}
The notations above are explained as follows. Here $C_{d,k}$ is the
normalization constant, given by
\begin{eqnarray*}
C_{d,k} =
\frac{\Gamma(dk)}{\prod^{d-1}_{j=0}\Gamma(k-j)\Gamma(d-j+1)}.
\end{eqnarray*}
Note that $\Gamma(z)$ is the Gamma function, defined for all complex
number with positive real part: $\Gamma(z) = \int^\infty_0
x^{z-1}e^{-x}\dif x$. In addition, $\delta$ is the Dirac delta
function \cite{Hoskins2009}. Besides, $\theta(x)$ is a function
defined by $\theta(x) = 1$ if $x>0$; otherwise, $\theta(x) = 0$. Now
we can answer the question \eqref{eq:equi-mix} based on the solution
to \eqref{eq:multi-orbits}. In fact, denote the spectral density of
$\rho_s$ by
$p\Pa{\boldsymbol{\lambda}^s|\boldsymbol{\lambda}^1,\ldots,\boldsymbol{\lambda}^n}$,
we conclude that the spectral density of $\rho$ in
\eqref{eq:equi-mix} is given by the following multiple integral:
$$
\int
p\Pa{\boldsymbol{\lambda}^s|\boldsymbol{\lambda}^1,\ldots,\boldsymbol{\lambda}^n}\prod^n_{j=1}p_{d,k}\Pa{\boldsymbol{\lambda}^j}\Br{\dif
\boldsymbol{\lambda}^j},
$$
where the function $p_{d,k}$ is from
\eqref{eq:pdk} and $\Br{\dif \boldsymbol{\lambda}^j }$ is the
Lebesgue volume element in $\real^d$, defined by $\Br{\dif
\boldsymbol{\lambda}^j } = \prod^d_{i=1}\dif \lambda^{(j)}_i$.

In view of this, let us focus on the derivation of the spectral
density of \eqref{eq:multi-orbits}. Before proceeding, some remarks
are necessary. Basically, what we are considered in the paper are
intimately related to Horn's problem (also called Horn's
conjecture), as mentioned previously. The Horn's problem asks for
the spectrum of eigenvalues of the sum of two given Hermitian
matrices with fixed eigenvalues. Specifically, given two Hermitian
matrices $\bsA$ and $\bsB$ with respective spectra $\bsa$ and
$\bsb$, the goal of Horn's problem is to identify the spectrum of
$\bsA+\bsB$. The solutions of Horn's problem, as vectors, form a
convex polytope whose describing linear inequalities have been
conjectured by Horn in 1962 \cite{Horn1962}. Although this problem
has already completely solved by Klyachko \cite{Klyachko1998} and
also by Knutson and Tao \cite{KT1999,KT2004}, later several other
proofs were presented, for instance, \cite{Alekseev2014}. Such
convex polytope are often defined by an exponential number of linear
inequalities when the matrix size are large enough. The general
problem is computationally intractable, it is also NP-hard. This
motivates us to get explicit computations for some small matrix
sizes. What we contributed in this paper is to derive explicit
expressions for the \emph{eigenvalue distribution} of the sum of
several matrices (say, two or more) instead of the determination of
the possible eigenvalues of the sum. Furthermore, these matrices are
restricted to be proportional to density matrices. Note that the
matrix size are restricted to no more than four in the paper since
analytical computation for matrix sizes larger than four seems too
complicated to be practical. In view of this reason, our method
provide a complete solution to the above problem algorithmically in
higher dimension.

We also consider the derivation of the probability density of the
diagonal part of $\rho_s$, defined in \eqref{eq:multi-orbits}. Let
us recall some notions related to it and its generalizations. A
well-known result of Isaac Schur \cite{Schur1923} indicates that the
diagonal elements $(a_{11},\ldots,a_{dd})$ of a $d\times d$
Hermitian matrix $\bsA=(a_{ij})$ solves a system of linear
inequalities involving the eigenvalues
$(\lambda_1,\ldots,\lambda_d)$. Indeed, if we viewing
$\bsa=(a_{11},\ldots,a_{dd})^\t$ and
$\boldsymbol{\lambda}=(\lambda_1,\ldots,\lambda_d)^\t$ as points in
$\real^d$, then by the Spectral Decomposition Theorem of a Hermitian
matrix, there exists a special unitary $\bsU\in\rS\rU(d)$ such that
$\bsA=\bsU\boldsymbol{\lambda}\bsU^\dagger$. Then $\bsa=(\bsU\star
\overline{\bsU})\boldsymbol{\lambda}$, where $\star$ stands for the
Schur product for both matrices with the same size, and the bar
means to taking the complex conjugate entrywise. That is, $\bsA\star
\bsB = (a_{ij}b_{ij})$. Clearly $\bsU\star \overline{\bsU}$ is a
$d\times d$ unistochastic matrix, a \emph{fortiori} bi-stochastic
matrix (with nonnegative entries, and both column-sum and row-sum
being equal to one). By Birkhoff-von Neumann theorem
\cite{Watrous2017}, we have obtained that $\bsa$ is in the convex
hull of the points $S_d\cdot\boldsymbol{\lambda}$. Here the action
of $S_d$ on the vector $\boldsymbol{\lambda}$ by permutating its
coordinates. Later, A. Horn shown \cite{Horn1954} the converse to
the above result holds true. Thus this convex hull is exactly the
set of diagonal parts of all elements from
$\cO_{\boldsymbol{\lambda}}$. A more general result related to
Horn's result is obtained by B. Kostant \cite{Kostant1974}. It is
easily seen from \eqref{eq:multi-orbits} that the diagonal part (as
a column vector) of $\rho_s$ in \eqref{eq:multi-orbits} is
determined by
\begin{eqnarray*}
\rho^D_s = \frac1n \sum^n_{j=1}(\bsU_j\star \overline{\bsU}_j)
\boldsymbol{\lambda}^j.
\end{eqnarray*}
The second goal of this paper is to identify the probability density
of $\rho^D_s$ in \eqref{eq:multi-orbits} when all
$\boldsymbol{\lambda}^j$ are fixed and $\bsU_j$ are
Haar-distributed.

The probability densities of the diagonal part and eigenvalues of
the mixture of several random qubits can be used to infer the
distribution of von Neumann entropy. As an application in quantum
information theory, we use our results to compute the average
entropy of the mixture of random quantum states and that of its
corresponding diagonal part. Furthermore, these computations can be
used to explain that the quantum coherence \cite{Baumgratz2014} of
the mixture decreases statistically as the number of components in
the mixture of qubits, as already noted in
\cite{Zhang2017a,Zhang2018}. In this process, we shall see that the
relative entropy of coherence, one of kinds of many quantum
coherence measures, defined via the relative entropy, naturally
relates the diagonal part and eigenvalues of a quantum state.

The paper is organized as follows. In
Section~\ref{sect:preliminary},  we present background tools related
to this paper, and recall the results obtained in \cite{Dooley1993}
by formulation used in \cite{Christandl2014}. Then, we consider the
equiprobable mixture of several qubit states, i.e., we derive the
spectral density of two qubit states and three qubit states
(Theorem~\ref{th:2}--Theorem~\ref{prop:3Qbit}), respectively, in
Section~\ref{sect:equi-mix}. Moreover, we present two examples to
demonstrate our method: In the first example, i.e.,
Example~\ref{exam:qutrit}, the density function of eigenvalues of
the equiprobable mixture of two qutrits with given spectra is
identified analytically; in the second example, i.e.,
Example~\ref{exam:qu4it}, we derive the density function of
eigenvalues of the equiprobable mixture of two two-qubits with given
spectra over some subregion of the support. Sequentially, in
Section~\ref{sect:app}, we present an application of our results in
quantum information theory. That is, the quantum coherence of the
mixture monotonously decreases statistically as the number of
components $n$. We conclude the paper with summary in
Section~\ref{sect:con-rem}.

\section{Preliminaries}\label{sect:preliminary}

As we shall see, Duistermaat-Heckman measure \cite{Guillemin2009} is
our central tool in the paper. In preparation for defining
Duistermaat-Heckman measure, we need to make an introduction about
push-forward measure or image measure from measure theory
\cite{Bogachev2007}. After that, we give the formal definition of
the moment map for a symplectic manifold \cite{McDuff2017}. Finally,
we focus on the product of coadjoint orbits where our problems will
be investigated. Note that we collect some notions (not new)
together in the section in order for the paper to be self-contained.

\subsection{Push-forward measure}

 In measure theory, a push-forward measure is
obtained by transferring a measure from one measurable space to
another using a measurable function. The following definition and
fact about push-forward measure can be found in \cite{Bogachev2007}.
\begin{definition}[Push-forward of a measure]
Given two measurable spaces $(\cX,\cF)$ and $(\cY,\cG)$, a
measurable mapping $\Phi:\cX\to\cY$ and a measure
$\mu:\cF\to[0,+\infty]$, the \emph{push-forward} of $\mu$ is defined
to be the measure $\Phi_*\mu: \cG\to[0,+\infty]$ given by
\begin{eqnarray*}
(\Phi_*\mu)(B) = \mu\Pa{\Phi^{-1}(B)}\quad \text{for}\quad B\in\cG.
\end{eqnarray*}
\end{definition}
The following result about the push-forward measure will be used in
the paper.
\begin{prop}[Change of variables formula]
Given two measurable spaces $(\cX,\cF)$ and $(\cY,\cG)$, a
measurable mapping $\Phi:\cX\to\cY$ and a measure
$\mu:\cF\to[0,+\infty]$. A measurable function $f$ on $\cY$ is
integrable with respect to the push-forward measure $\Phi_*\mu$ if
and only if the composition $\Phi\circ f$ is integrable with respect
to the measure $\mu$. In that case, the integrals coincide:
\begin{eqnarray*}
\int_{\cY} f\dif(\Phi_*\mu) = \int_{\cX} f\circ\Phi
\dif\mu:=\int_{\cX} (\Phi^*f)\dif\mu,
\end{eqnarray*}
where $\Phi^*f:=f\circ\Phi$ is the function on $\cX$, called the
pull-back of the function $f$ on $\cY$. In the notation of
distribution, the above fact can be represented by
\begin{eqnarray*}
\Inner{\Phi_*\mu}{f} = \Inner{\mu}{\Phi^*f}.
\end{eqnarray*}
Here $\Inner{\cdot}{\cdot}$ denotes the pairing between measures and
test functions.
\end{prop}
With this notion, we can describe the main notions used in the
paper.

\subsection{Moment map}

The notions mentioned in this part can be found in
\cite{Marsden1978,Kirillov}.

\begin{definition}[Symplectic
manifold] Assume that $M$ is a smooth manifold. $M$ is called
\emph{symplectic manifold} if there exists a closed non-degenerate
2-form $\omega_M$ on $M$. That is,
\begin{enumerate}[(i)]
\item \textbf{$\omega_M$ is a 2-form:} It is a anti-symmetric and
bilinear form on the product of two tangent spaces
$\T_mM\times\T_mM$ for each $m\in M$;
\item \textbf{$\omega_M$ is closed:} $\dif\omega_M=0$;
\item \textbf{$\omega_M$ is non-degenerate:} on each tangent space $\T_mM(m\in M)$: if
$\omega_M(\xi,\eta)=0$ for all $\eta\in\T_mM$, then $\xi=0$.
\end{enumerate}
Note that a closed non-degenerate 2-form is called \emph{symplectic
form}.
\end{definition}

\begin{definition}[Action of a Lie group on a manifold]
Let $M$ be a smooth manifold. An \emph{action} of a Lie group $G$ on
$M$ is a smooth mapping $\tau: G\times M\to M$, such that
\begin{enumerate}[(i)]
\item for all $m\in M$, $\tau(e,m)=m$ and
\item for every $g,h\in G$, $\tau(g,\tau(h,m))=\tau(gh,m)$ for all
$m\in M$.
\end{enumerate}
\end{definition}
For every $g\in G$, let $\tau_g: M\to M$ be given via $m\mapsto
\tau_g(m)$. The above definition can be rephrased as: the mapping
$g\mapsto \tau_g$ is a homomorphism of $G$ into the group of
diffeomorphisms of $M$. If $M$ is a \emph{vector space} and each
$\tau_g$ is a linear transformation, the action of $G$ on $M$ is
called a \emph{representation} of $G$ on $M$.

Suppose $\tau:G\times M\to M$ is a smooth action. If $\xi\in\g$, the
Lie algebra of $G$, then $\tau^\xi: \real\times M\to M$ is given via
$(t,\xi)\mapsto \tau(e^{t\xi},m)$ is an $\real$-action on $M$, that
is, $\tau^\xi$ is a flow on $M$. The corresponding vector field on
$M$ is given by $\xi_M(m):=\left.\frac{\dif}{\dif t}\right|_{t=0}
\tau_{e^{t\xi}}(m)$ is called the \emph{infinitesimal generator} of
the action corresponding to $\xi$.

The \emph{adjoint representation} $\Ad$ of $G$ on its Lie algebra
$\g$ is given by
\begin{eqnarray*}
G\times \g\to\g, \quad
(g,\xi)\mapsto g\cdot\xi:=\Ad(g)\xi=g\xi g^{-1}.
\end{eqnarray*}
This adjoint representation of $G$ on $\g$ induces a \emph{coadjoint
representation} $\Ad^*$ of $G$ on $\g^*$ (the dual space of $\g$):
\begin{eqnarray*}
G\times\g^*\to\g^*,\quad (g,\varphi)\mapsto g\cdot\varphi =
\Ad^*(g)\varphi=\Ad(g^{-1})^*\varphi,
\end{eqnarray*}
where
$\Inner{\Ad(g^{-1})^*\varphi}{\eta}=\Inner{\varphi}{\Ad(g^{-1})\eta}$.
That is, $\Ad(g^{-1})^*\varphi$ is the pull-back of $\varphi$ under
the mapping $\Ad(g^{-1})$.

\begin{definition}[Moment map]\label{def:momentmap}
Let $(M,\omega_M)$ be a connected symplectic manifold and
$\tau:G\times M\to M$ a \emph{symplectic action} of the Lie group
$G$ on $M$; that is, for each $g\in G$, the map $\tau_g:M\to M$ is
given via $\tau_g(m)=\tau(g,m)$ is symplectic, i.e.,
$\tau^*_g\omega_M=\omega_M$. We say that a mapping $\Phi_G:
M\to\g^*$ is a \emph{moment map} for the action if, for every
$\xi\in\g$,
\begin{eqnarray}\label{eq:exact-form-for-momentmap}
\dif\phi^\xi = \omega_M(\xi_M,\cdot):= \iota(\xi_M)\omega_M,
\end{eqnarray}
where $\Inner{\Phi_G}{\xi}=\phi^\xi: M\to\real$ is defined by
$\phi^\xi(m):=\Inner{\Phi_G(m)}{\xi}$, and $\xi_M$ is the
infinitesimal generator of the action corresponding to $\xi\in\g$.
Sometimes $(M,\omega_M,G,\Phi_G)$ is called a \emph{Hamiltonian
$G$-manifold}.
\end{definition}

\begin{exam}[The symplectic structure of a co-adjoint orbit]
Let $G$ be a compact connected Lie group with its Lie algebra $\g$.
Assume that $G$ acts on $\g$ by the adjoint action, and acts on
$\g^*$ by the co-adjoint action, as mentioned above. Fix
$\lambda\in\g^*$ and let $\cO_\lambda=G\cdot\lambda$ by the
co-adjoint action. The infinitesimal generator of the co-adjoint
action of $G$, corresponding to $\xi\in\g$, is given by
$\xi_{\cO_{\lambda}} (f) = - f\circ \ad(\xi),\quad f\in\cO_\lambda$.
The Kirillov-Kostant-Souriau 2-form on $\cO_\lambda$ is defined by
\begin{eqnarray}\label{eq:KKS-form}
\omega(\xi_{\cO_\lambda}(f), \eta_{\cO_\lambda}(f)) = -
f([\xi,\eta]),\quad f\in\cO_\lambda,\xi,\eta\in\g.
\end{eqnarray}
Moreover $\omega$ is a well-defined smooth and closed non-degenerate
2-form on $\cO_\lambda$ so that $(\cO_\lambda,\omega)$ is a
symplectic manifold. Furthermore, $\omega$ is $G$-invariant, thus a
symplectic form. The minus inclusion map
$\Phi_G:\cO_\lambda\hookrightarrow\g^*$ is a moment map, i.e.,
$\Phi_G$ is $G$-equivariant and $\dif\phi^\xi =
\iota(\xi_{\cO_\lambda})\omega\quad \xi\in\g$, where
$\phi^\xi=\Inner{\Phi_G}{\xi}:\cO_\lambda\to\real$. Note that
adjoint orbits in $\g$ can be identified with co-adjoint orbits in
$\g^*$ by choosing a positive definite $G$-invariant form on $\g$.
This convention will be used throughout in the paper.
\end{exam}

\subsection{Duistermaat-Heckman measure}

Throughout the paper, $K$ will denote a compact, connected Lie group
with maximal torus (maximal commutative subgroup) $T\subset K$, Weyl
group (the normalizer of $T$) $W$, respective Lie algebras
$\k,\liet$. We write $\pi_{K,T}:\k^*\to\liet^*$ for the projection
dual to the inclusion $\liet\subset\k$. Here $\k^*$ ($\liet^*$)
means the dual space of $\k$ ($\liet$). Let us also choose a
positive Weyl chamber $\liet^*_{\geqslant0}\subset\liet^*$; this
determines a set of positive roots
$\set{\alpha_1,\ldots,\alpha_R}\subset\mathrm{i}\liet^*$. All
positive roots are denoted simple by $\alpha>0$. We also denote by
$\liet^*_{>0}$ the interior of the positive Weyl chamber. The
notions related to theory of compact Lie groups and its Lie algebras
can be found in \cite{Hall2015}.

\begin{definition}[Duistermaat-Heckman measure]
Let $M$ be a compact, connected Hamiltonian $K$-manifold of
dimension $2m$, with symplectic form $\omega_M$ (a closed
non-degenerate 2-form) and a choice of moment map $\Phi_K:
M\to\mathrm{i}\k^*$, as in Definition~\ref{def:momentmap}. The
Liouville measure on $M$ is defined by $\mu_M:=\frac{\omega^{\wedge
m}_M}{(2\pi)^mm!}$, where $\omega^{\wedge
m}_M=\overbrace{\omega_M\wedge\cdots\wedge \omega_M}^m$ is the
$m$-th exterior power of $\omega_M$. More precisely, for a Borel
subset $B$ of $M$, the Liouville measure of $B$ is given by
$\mu_M(B)=\int_B\frac{\omega^{\wedge m}_M}{(2\pi)^mm!}$. The
non-Abelian \emph{Duistermaat-Heckman measure} $\dh^K_M$ is defined
as follows:
\begin{eqnarray*}
\dh^K_M = \frac1{p_K}(\tau_K)_*(\Phi_K)_*(\mu_M),
\end{eqnarray*}
where $\tau_K:\mathrm{i}\k^*\to\mathrm{i}\liet^*_{\geqslant0}$ is
defined as $\tau_K(\cO_\lambda)=\lambda$ for $\lambda\in
\mathrm{i}\liet^*_{\geqslant0}$ and
$p_K(\lambda)=\vol(\cO_\lambda)=\prod_{\alpha>0}\frac{\inner{\lambda}{\alpha}}{\inner{\alpha}{\varpi}}$
for $\varpi=\frac12\sum_{\alpha>0}\alpha$, half the sum of all
positive roots. Here $\vol(\cO_\lambda)$ is the symplectic volume of
such co-adjoint orbit $\cO_\lambda$ of dimension $2m$, it is
specifically given by $\vol(\cO_\lambda) =
\int_{\cO_\lambda}\frac{\omega^{\wedge
m}_{\cO_\lambda}}{(2\pi)^mm!}$, where the definition of
$\omega_{\cO_\lambda}$ is taken from \eqref{eq:KKS-form}.
\end{definition}
Thus the D-H measure associated with the co-adjoint action of $K$ on
a generic co-adjoint orbit $\cO_\lambda$ is a probability
distribution concentrated at the point $\lambda$. The \emph{Abelian}
D-H measure is defined as:
\begin{eqnarray*}
\dh^T_M = (\pi_{K,T})_*(\Phi_K)_*(\mu_M) = (\Phi_T)_*(\mu_M),
\end{eqnarray*}
where $\Phi_T:M\to\mathrm{i}\liet^*$ is given by
$\Phi_T=\pi_{K,T}\circ\Phi_K$. That is, the Liouville measure
$\mu_M$ is pushed forward along the moment map $\Phi_K$, and further
pushed forward along the map $\pi_{K,T}$. In more detail, for a
Borel subset $\cB$ in $\mathrm{i}\liet^*$, $\dh^T_M(\cB) =
\int_{\Phi^{-1}_T(\cB)} \frac{\omega^{\wedge m}_M}{(2\pi)^mm!}$.

\subsection{The product of co-adjoint orbits}

Firstly, we recall the following general fact which can be found in
the literature.
\begin{prop}
The product of symplectic manifolds $(M_1,\omega_1)$ and
$(M_2,\omega_2)$ is a symplectic manifold with respect to the form
$a_1\cdot p^*_1\omega_1+a_2\cdot p^*_2\omega_2$ for nonzero real
numbers $a_1,a_2\in\real$. Here $p_i:M_1\times M_2\to M_i$ is the
projection, where $i=1,2$.
\end{prop}

Consider the diagonal co-adjoint action of $K$ on the manifold
$M=\cO_{\lambda_1}\times\cO_{\lambda_2}$, where
$\lambda_1,\lambda_2\in\mathrm{i}\liet^*_{\geqslant0}$, which is
given via $h\cdot (f_1,f_2):=(\Ad^*(h)f_1,\Ad^*(h)f_2)$ for any
$h\in K$ and any $(f_1,f_2)\in
\cO_{\lambda_1}\times\cO_{\lambda_2}$. Let $\omega_1$ and $\omega_2$
be the Kirillov-Kostant-Souriau 2-forms defined over
$\cO_{\lambda_1}$ and $\cO_{\lambda_2}$, respectively. Denote by
$p_i:\cO_{\lambda_1}\times\cO_{\lambda_2}\to\cO_{\lambda_i} (i=1,2)$
be the projections. We can take $p^*_1\omega_1+p^*_2\omega_2$ as the
symplectic form $\omega$ on
$M=\cO_{\lambda_1}\times\cO_{\lambda_2}$. Let
$\lambda_1,\lambda_2\in\mathrm{i}\liet^*_{>0}$. Define the moment
map as follows: $\Phi_K:
\cO_{\lambda_1}\times\cO_{\lambda_2}\to\mathrm{i}\k^*,\quad
\Phi_K(f_1,f_2)=-(f_1+f_2)$. Clearly the map $\Phi_K$ is
$K$-equivariant in the sense that $\Phi_K(h\cdot(f_1,f_2)) = h\cdot
\Phi_K(f_1,f_2), \forall h\in K$. Next we check that $\Phi_K$
satisfies \eqref{eq:exact-form-for-momentmap}. Indeed, the
infinitesimal generator corresponding to $\xi\in\k$ is given by
$\xi_{\cO_{\lambda_1}\times\cO_{\lambda_2}} =
(\xi_{\cO_{\lambda_1}},\xi_{\cO_{\lambda_2}})$. Furthermore, $
(p_i)_*\xi_{\cO_{\lambda_1}\times\cO_{\lambda_2}} =
(p_i)_*(\xi_{\cO_{\lambda_1}},\xi_{\cO_{\lambda_2}})=
\xi_{\cO_{\lambda_i}}\quad i=1,2$. For $\xi,\eta\in\k$, we have
their infinitesimal generators on the product
$\cO_{\lambda_1}\times\cO_{\lambda_2}$ are
$\xi_{\cO_{\lambda_1}\times\cO_{\lambda_2}}$ and
$\eta_{\cO_{\lambda_1}\times\cO_{\lambda_2}}$, respectively. Denote
$\Inner{\Phi_K}{\xi}:=\phi^\xi$. Then, on the one hand,
\begin{eqnarray*}
\dif\phi^\xi(\eta_{\cO_{\lambda_1}\times\cO_{\lambda_2}}) =
\left.\frac{\dif}{\dif
t}\right|_{t=0}\phi^\xi\Pa{e^{t\eta}\cdot(f_1,f_2)}=
-(\Inner{f_1}{[\xi,\eta]}+\Inner{f_2}{[\xi,\eta]}).
\end{eqnarray*}
where $(f_1,f_2)\in \cO_{\lambda_1}\times\cO_{\lambda_2}$. That is,
\begin{eqnarray}\label{eq:differential-lhs}
\dif\phi^\xi(\eta_{\cO_{\lambda_1}\times\cO_{\lambda_2}}) =
\omega_1(\xi_{\cO_{\lambda_1}}(f_1),\eta_{\cO_{\lambda_1}}(f_1))+
\omega_2(\xi_{\cO_{\lambda_2}}(f_2),\eta_{\cO_{\lambda_2}}(f_2)).
\end{eqnarray}
On the other hand,
$(\iota(\xi_{\cO_{\lambda_1}\times\cO_{\lambda_2}})\omega)(\eta_{\cO_{\lambda_1}\times\cO_{\lambda_2}})
=
\omega(\xi_{\cO_{\lambda_1}\times\cO_{\lambda_2}}(f),\eta_{\cO_{\lambda_1}\times\cO_{\lambda_2}}(f))$,
where $f=(f_1,f_2)\in \cO_{\lambda_1}\times\cO_{\lambda_2}$. Since
$\omega=p^*_1\omega_1+p^*_2\omega_2$, it follows that
\begin{eqnarray*}
\omega(\xi_{\cO_{\lambda_1}\times\cO_{\lambda_2}}(f),\eta_{\cO_{\lambda_1}\times\cO_{\lambda_2}}(f))
=\omega_1(\xi_{\cO_{\lambda_1}}(f_1),\eta_{\cO_{\lambda_1}}(f_1))+\omega_2(\xi_{\cO_{\lambda_2}}(f_2),\eta_{\cO_{\lambda_2}}(f_2)).
\end{eqnarray*}
Thus
\begin{eqnarray}\label{eq:differential-rhs}
(\iota(\xi_{\cO_{\lambda_1}\times\cO_{\lambda_2}})\omega)(\eta_{\cO_{\lambda_1}\times\cO_{\lambda_2}})
=\omega_1(\xi_{\cO_{\lambda_1}}(f_1),\eta_{\cO_{\lambda_1}}(f_1))+\omega_2(\xi_{\cO_{\lambda_2}}(f_2),\eta_{\cO_{\lambda_2}}(f_2)).
\end{eqnarray}
By combining \eqref{eq:differential-lhs} and
\eqref{eq:differential-rhs}, we see that $\dif\Inner{\Phi_K}{\xi} =
\iota(\xi_{\cO_{\lambda_1}\times\cO_{\lambda_2}})\omega$.

Now we can consider the problem of describing the sum of two
coadjoint orbits (also equivalently identified with adjoint orbits
under the adjoint action of the special unitary group)
$\cO_{\lambda}+\cO_{\mu}$. This is the so-called Horn's problem. Let
$\lambda\in\mathrm{i}\liet^*_{>0}$ and
$\mu\in\mathrm{i}\liet^*_{\geqslant0}$. Then
\cite{Dooley1993,Christandl2014},
\begin{eqnarray}\label{eq:NA-DH-measure}
\dh^K_{\cO_\lambda\times\cO_\mu} = \sum_{w\in
W}(-1)^{l(w)}\delta_{w\lambda}*\dh^T_{\cO_\mu}
\end{eqnarray}
where $l(w)$ is the length of the Weyl group element $w$, and
$\delta_{\alpha}$ for the Dirac measure at $\alpha$; $*$ means the
convolution, the same below. Moreover, we also have the following
result \cite{Christandl2014}:
\begin{eqnarray}\label{eq:abelian-dh-measure-over-orbit}
\dh^T_{\cO_\mu} = \sum_{w\in W}(-1)^{l(w)}\delta_{w\mu}*
H_{-\alpha_1}*\cdots* H_{-\alpha_R},
\end{eqnarray}
where $H_{\omega}$ is the so-called \emph{Heaviside measure} which
is defined by $\Inner{H_\omega}{f} = \int^\infty_0
f(t\cdot\omega)\dif t$. Both results, i.e., \eqref{eq:NA-DH-measure}
and \eqref{eq:abelian-dh-measure-over-orbit}, will be employed to
derive the eigenvalue density of the mixture of several qubit states
in this paper. We summarize the above results into the following
proposition. Note that
$\delta_\alpha*\delta_\beta=\delta_{\alpha+\beta}$.
\begin{prop}
Let $\lambda\in\mathrm{i}\liet^*_{>0}$ and
$\mu,\nu\in\mathrm{i}\liet^*_{\geqslant0}$. Then
\begin{eqnarray}
\dh^K_{\cO_\lambda\times\cO_\mu} &=& \sum_{w,w'\in
W}(-1)^{l(w)+l(w')}\delta_{w\lambda+w'\mu}* H_{-\alpha_1}*\cdots*
H_{-\alpha_R},\\
\dh^K_{\cO_\lambda\times\cO_\mu\times\cO_\nu} &=& \sum_{w,w',w"\in
W}(-1)^{l(w)+l(w')+l(w")}\delta_{w\lambda+w'\mu+w"\nu}*
H_{-\alpha_1}*\cdots* H_{-\alpha_R}.
\end{eqnarray}
\end{prop}
Of course, we can generalize further the above proposition to
compute the non-Abelian Duistermaat-Heckman measure
$\dh^K_{\cO_\lambda\times\cO_{\mu_1}\times\cdots\times\cO_{\mu_q}}$,
$\lambda\in\mathrm{i}\liet^*_{>0}$ and
$\mu_1,\ldots,\mu_q\in\mathrm{i}\liet^*_{\geqslant0}$. But this not
the goal of this paper.

Next let us focus on the case where $K=\rS\rU(d)$, the set of all
$d\times d$ unitary matrices with unit determinant, its Lie algebra
is $\k=\su(d)$, the set of all $d\times d$ skew-Hermitian matrices
with trace zero, and its maximal torus $T$, the maximal commutative
subgroup of $K$, with its Lie algebra $\liet$ being identified with
the set of all $d\times d$ diagonal matrices with imaginary entries
of zero trace. The Weyl group $W$ of $\rS\rU(d)$ is $S_d$ (up to
isomorphism). In such case, theoretically, we can do analytical
computation about Duistermaat-Heckman measure albeit this problem
has exponential complexity. With the previous preparation, in the
next section, we derive some explicit expressions for qubit
situations.

\section{Main results}\label{sect:equi-mix}

In the following two subsections, let $K=\rS\rU(d)$ for $d=2$. Then
$T=\Set{\diag(e^{\mathrm{i}\theta},e^{-\mathrm{i}\theta}):\theta\in\real}$
with its Lie algebra
$\liet=\Set{\diag(\mathrm{i}\theta,-\mathrm{i}\theta):\theta\in\real}$.
Thus $\bsh=\diag(1,-1)$ is the basis of $\liet$, and $\liet=
\mathrm{i}\real \cdot \bsh\cong\real$. Its unique positive root
$\alpha\in\liet^*$ is given by
\begin{eqnarray}\label{eq:bsh}
\alpha(\bsh)=2,\quad \text{ where } \bsh=\diag(1,-1).
\end{eqnarray}
Hence $\alpha\cong \bsh$ via the Hilbert-Schmidt inner product
$\alpha(\bsh)=\inner{\bsh}{\bsh}$. Besides, the Weyl group is given
by $W\cong S_2$, the permutation group of degree 2. Let
$\lambda\in\mathrm{i}\liet^*_{>0}$. Then
\eqref{eq:abelian-dh-measure-over-orbit} reduces to the following
\begin{eqnarray*}
\dh^T_{\cO_\lambda} = \sum_{w\in S_2}(-1)^{l(w)}\delta_{w\lambda}*
H_{-\alpha}.
\end{eqnarray*}
\subsection{The mixture of two qubit states}

For the mixture of two qubit states, we have that the Abelian D-H
measure over the manifold $\cO_\mu\times\cO_\nu$ is given by the
following convolution: $\dh^T_{\cO_\mu\times\cO_\nu} =
\dh^T_{\cO_\mu}*\dh^T_{\cO_\nu}$. Thus the non-Abelian D-H measure
is the following: $\dh^K_{\cO_\mu\times\cO_\nu} = \sum_{w\in
W}(-1)^{l(w)}\delta_{w\mu}*\dh^T_{\cO_\nu}$, where
$\mu\in\mathrm{i}\liet^*_{>0}$ is represented by $\mu\cong \mu\cdot
\bsh$ for $\mu>0$. Now, for $\mu>0,\nu\geqslant0$,
\begin{eqnarray*}
\dh^K_{\cO_\mu\times\cO_\nu} = \left.\Pa{\delta_{(\mu+\nu)\cdot
\bsh} + \delta_{-(\mu+\nu)\cdot \bsh} - \delta_{(\mu-\nu)\cdot \bsh}
- \delta_{(\nu-\mu)\cdot \bsh}}*
H_{-\bsh}\right|_{\mathrm{i}\liet^*_{\geqslant0}}.
\end{eqnarray*}
Therefore, we see that $\dh^K_{\cO_\mu\times\cO_\nu} =
\Pa{\delta_{(\mu+\nu)\cdot \bsh}- \delta_{\abs{\mu-\nu}\cdot \bsh}}*
H_{-\bsh}$. Furthermore,
\begin{eqnarray*}
\Inner{\dh^K_{\cO_\mu\times\cO_\nu}}{f} = \int^{\mu+\nu}_0
f(t\cdot\bsh)\dif t - \int^{\abs{\mu-\nu}}_0 f(t\cdot\bsh)\dif t=
\int^{\mu+\nu}_{\abs{\mu-\nu}} f(t\cdot\bsh)\dif t.
\end{eqnarray*}
According to the definition of D-H measure $\dh^K_M =
\frac1{p_K}(\tau_K)_*(\Phi_K)_*\mu_M$, where $K=\rS\rU(2)$ and
$M=\cO_\mu\times\cO_\nu$. Multiplying the non-Abelian
Duistermaat-Heckman measure by the symplectic volume polynomial
$p_K(\lambda\cdot\bsh)=2\lambda$, thus we have
\begin{eqnarray*}
p_K\frac{\dh^K_M}{\vol(M)} =
(\tau_K)_*(\Phi_K)_*\Pa{\frac{\mu_M}{\vol(M)}},
\end{eqnarray*}
where $\vol(M) = \vol(\cO_\mu)\vol(\cO_\nu)=4\mu\nu$. Next,
\begin{eqnarray*}
\Inner{(\tau_K)_*(\Phi_K)_*\Pa{\frac{\mu_M}{\vol(M)}}}{f}
=\frac1{4\mu\nu}\Inner{\dh^K_M}{p_Kf}=
\frac1{2\mu\nu}\int^{\mu+\nu}_{\abs{\mu-\nu}}
f(\lambda\cdot\bsh)\lambda\dif \lambda.
\end{eqnarray*}
That is, the density of $\lambda$ with respect to the measure
$(\tau_K)_*(\Phi_K)_*\Pa{\frac{\mu_M}{\vol(M)}}$ is given by the
following analytical formula:
\begin{eqnarray}\label{eq:mix-2-qubit}
p(\lambda) = \frac{\lambda}{2\mu\nu},\quad
\lambda\in\Br{\abs{\mu-\nu},\mu+\nu}.
\end{eqnarray}
A direct consequence of \eqref{eq:mix-2-qubit} can be obtained
immediately.
\begin{prop}\label{prop:lmn}
The probability density function of an eigenvalue $s$ of the mixture
$\rho^w=w\rho_1+(1-w)\rho_2 (w\in(0,1))$ of two random density
matrices, chosen uniformly from respective unitary orbits $\cO_a$
and $\cO_b$ with $a,b$ are fixed in $\Pa{0,\frac12}$, i.e.,
$\cO_a:=\cO_{(1-a,a)}$ and $\cO_b:=\cO_{(1-b,b)}$, is given by
\begin{eqnarray}\label{eq:density-eigen}
f_w(s|a,b) =
\frac1{4w(1-w)}\times\frac{\abs{s-\frac12}}{\Pa{\frac12-a}\Pa{\frac12-b}},
\end{eqnarray}
where $s \in[t_0(w),t_1(w)]\cup[1-t_1(w),1-t_0(w)]$. Here
\begin{eqnarray}\label{eq:t0-t1-of-w}
t_0(w):=wa+(1-w)b,\quad
t_1(w):=\frac12-\Abs{w\Pa{\frac12-a}-(1-w)\Pa{\frac12-b}}.
\end{eqnarray}
\end{prop}

\begin{proof}
In fact, let $\rho_a\in
\cO_a:=\Set{\bsU\diag(1-a,a)\bsU^\dagger:\bsU\in\rS\rU(2)}$ and
$\rho_b\in\cO_b$, consider the mixture of $\rho_a$ and $\rho_b$,
i.e., $\rho^w_s = w\rho_a+(1-w)\rho_b$ with $a,b\in\Pa{0,\frac12}$.
Then let $\mu=w\Pa{\frac12-a},\nu=(1-w)\Pa{\frac12-b}$ and
$\lambda=\Abs{\frac12-s}$ for $s\in[0,1]$. Then
$\lambda\in\Br{\abs{\mu-\nu},\mu+\nu}$ can be expressed as
$$
wa+(1-w)b\leqslant s\leqslant \min\Pa{wa+(1-w)(1-b),w(1-a)+(1-w)b}.
$$
Note that
$$
\min\Pa{wa+(1-w)(1-b),w(1-a)+(1-w)b}=\frac12-\Abs{w\Pa{\frac12-a}-(1-w)\Pa{\frac12-b}}.
$$
Therefore we have that
\begin{eqnarray*}
f_w(s|a,b) =\frac1{4w(1-w)}\times
\frac{\abs{s-\frac12}}{\Pa{\frac12-a}\Pa{\frac12-b}},
\end{eqnarray*}
where $s \in[t_0(w),t_1(w)]\cup[1-t_1(w),1-t_0(w)]$. This completes
the proof.
\end{proof}
This extends the result obtained in \cite{Zhang2017c}. Recall that
any qubit density matrix can be represented as
\begin{eqnarray}\label{eq:Bloch-rep}
\rho(\bsr)=\frac12(\I_2+\bsr\cdot\boldsymbol{\sigma}),
\end{eqnarray}
where $\bsr=(r_x,r_y,r_z)\in\real^3$ is the Bloch vector with its
length $r:=\abs{\bsr}\leqslant1$, and
$\boldsymbol{\sigma}=(\sigma_x,\sigma_y,\sigma_z)$, where
\begin{eqnarray*}
\sigma_x=\Pa{\begin{array}{cc}
               0 & 1 \\
               1 & 0
             \end{array}
},\quad \sigma_y=\Pa{\begin{array}{cc}
                       0 & -\mathrm{i} \\
                       \mathrm{i} & 0
                     \end{array}
},\quad\sigma_z=\Pa{\begin{array}{cc}
                      1 & 0 \\
                      0 & -1
                    \end{array}
}
\end{eqnarray*}
are three Pauli matrices. The relationship between both eigenvalues
of a qubit density matrix and the length of its corresponding Bloch
vector is given by: $\lambda_{\pm}(\rho(\bsr))=\frac12(1\pm r)$.
Note also that the probability density for the length $r$ of the
Bloch vector $\bsr$ in the Bloch representation \eqref{eq:Bloch-rep}
of a random qubit $\rho\in\cE_{2,2}$, i.e., by partial-tracing over
a Haar-distributed pure two-qubit state is given by
\cite{Zhang2018}:
\begin{eqnarray}\label{eq:Bloch-lenghth}
p^{(1)}(r) = 3r^2, \quad r\in[0,1].
\end{eqnarray}
Then \eqref{eq:density-eigen} can be reformulated as the following
form.
\begin{cor}\label{cor:rr1r2}
The conditional probability density function of the length $r$ of
Bloch vector $\boldsymbol{r}$ of the mixture:
$\rho_w(\bsr)=w\rho(\bsr_1)+(1-w)\rho(\bsr_2)(w\in(0,1))$, where
$r_1,r_2\in(0,1)$ are fixed, is given by
\begin{eqnarray}\label{eq:rr1r2}
p_w(r|r_1,r_2) = \frac1{2w(1-w)}\times\frac{r}{r_1r_2},
\end{eqnarray}
where $r\in[r^-_w,r^+_w]$. Here $r^-_w:=\abs{wr_1-(1-w)r_2}$ and
$r^+_w:=wr_1+(1-w)r_2$.
\end{cor}
In particular, for $w=\frac12$, the above result in
Corollary~\ref{cor:rr1r2} is reduced to the one obtained in
\cite{Zhang2018}. From Probability Theory, we know that
\begin{eqnarray*}
p^{(2)}_w(r) =
\iint_{R_w(r)}p_w(r|r_1,r_2)p^{(1)}(r_1)p^{(1)}(r_2)\dif r_1\dif
r_2,
\end{eqnarray*}
where $R_w(r)$ is a family of sections, parameterized by
$r\in[0,1]$, where arbitrary $w\in(0,1)$ is fixed:
\begin{eqnarray*}
R_w(r) :=\Set{(r_1,r_2)\in[0,1]^2: \abs{wr_1-(1-w)r_2}\leqslant
r\leqslant wr_1+(1-w)r_2},\quad r\in[0,1].
\end{eqnarray*}
Without loss of generality, we assume that $w\geqslant 1-w$, i.e.,
$w\geqslant\frac12$. If $w=\frac12$, then we will drop those
corresponding subindexes $w$ of the quantities
$p^{(2)}_w(r),R_w(r),p_w(r|r_1,r_2)$ and they reduces as
$p^{(2)}(r),R(r),p(r|r_1,r_2)$.

\begin{thrm}\label{th:2}
The probability density function $p^{(2)}_w$ of the length $r$ of
the Bloch vector $\bsr$ of the mixture:
$\rho_w(\bsr)=w\rho_1+(1-w)\rho_2(w\in\Pa{\frac12,1})$, where
$\rho_j\in\cE_{2,2}(j=1,2)$, is given by
\begin{eqnarray}\label{eq:pdf2}
p^{(2)}_w(r) = \begin{cases}\frac{3(1-r)^2r\Pa{r^2+2r-12w^2+12w-3}}{16w^3(1-w)^3},& r\in[2w-1,1],\\
\frac{3r^2}{w^3},& r\in\Br{0,2w-1}.\end{cases}
\end{eqnarray}
Moreover, we have that $\lim_{w\nearrow 1}p^{(2)}_w(r)=3r^2\text{
and } \lim_{w\searrow \frac12}p^{(2)}_w(r)=12r^2(r^3-3r+2)$.
\end{thrm}

\begin{proof}
Without loss of generality, we fix arbitrary $w\in\Pa{\frac12,1}$.
In the following discussion, we omit the integrand for simplicity.\\
(i) For $w\in\Pa{\frac23,1}$, $w>2w-1>1-w$. Then we see that \\
(i1) If $r\in[w,1]$, then
\begin{eqnarray*}
p^{(2)}_w(r) =\int^1_{1+\frac{r-1}w}\dif
r_1\int^1_{\frac{r-wr_1}{1-w}}\dif r_2=
\frac{3(1-r)^2r\Pa{r^2+2r-12w^2+12w-3}}{16w^3(1-w)^3}.
\end{eqnarray*}
(i2) If $r\in[2w-1,w]$, then
\begin{eqnarray*}
p^{(2)}_w(r) =\int^{\frac rw}_{1+\frac{r-1}w}\dif
r_1\int^1_{\frac{r-wr_1}{1-w}}\dif r_2+\int^1_{\frac rw}\dif
r_1\int^1_{\frac{wr_1-r}{1-w}}\dif r_2=
\frac{3(1-r)^2r\Pa{r^2+2r-12w^2+12w-3}}{16w^3(1-w)^3}.
\end{eqnarray*}
(i3) If $r\in[1-w,2w-1]$, then
\begin{eqnarray*}
p^{(2)}_w(r) = \int^{\frac rw}_{1+\frac{r-1}w}\dif
r_1\int^1_{\frac{r-wr_1}{1-w}}\dif r_2+\int^{\frac{1+r}w-1}_{\frac
rw}\dif r_1\int^1_{\frac{wr_1-r}{1-w}}\dif r_2 =\frac{3r^2}{w^3}.
\end{eqnarray*}
(i4) If $r\in\Br{\frac{1-w}2,1-w}$, then
\begin{eqnarray*}
p^{(2)}_w(r) = \int^{\frac{1-r}w-1}_0\dif
r_1\int^{\frac{r+wr_1}{1-w}}_{\frac{r-wr_1}{1-w}}\dif
r_2+\int^{\frac rw}_{\frac{1-r}w-1}\dif
r_1\int^1_{\frac{r-wr_1}{1-w}}\dif r_2+\int^{\frac{1+r}w-1}_{\frac
rw}\dif r_1\int^1_{\frac{wr_1-r}{1-w}}\dif r_2 = \frac{3r^2}{w^3}.
\end{eqnarray*}
(i5) If $r\in\Br{0,\frac{1-w}2}$, then
\begin{eqnarray*}
p^{(2)}_w(r) =\int^{\frac rw}_0\dif
r_1\int^{\frac{r+wr_1}{1-w}}_{\frac{r-wr_1}{1-w}}\dif r_2 +
\int^{\frac{1-r}w-1}_{\frac rw}\dif
r_1\int^{\frac{wr_1+r}{1-w}}_{\frac{wr_1-r}{1-w}}\dif r_2 +
\int^{\frac{1+r}w-1}_{\frac{1-r}w-1}\dif
r_1\int^1_{\frac{wr_1-r}{1-w}}\dif r_2= \frac{3r^2}{w^3}.
\end{eqnarray*}
(ii) For $w\in\left(\frac12,\frac23\right]$, $w>1-w\geqslant2w-1>\frac{1-w}2$. Then we see that \\
(ii-1) If $r\in[w,1]$, then
\begin{eqnarray*}
p^{(2)}_w(r) =\int^1_{1+\frac{r-1}w}\dif
r_1\int^1_{\frac{r-wr_1}{1-w}}\dif r_2=
\frac{3(1-r)^2r\Pa{r^2+2r-12w^2+12w-3}}{16w^3(1-w)^3}.
\end{eqnarray*}
(ii-2) If $r\in[1-w,w]$, then
\begin{eqnarray*}
p^{(2)}_w(r) =\int^{\frac rw}_{1+\frac{r-1}w}\dif
r_1\int^1_{\frac{r-wr_1}{1-w}}\dif r_2+\int^1_{\frac rw}\dif
r_1\int^1_{\frac{wr_1-r}{1-w}}\dif
r_2=\frac{3(1-r)^2r\Pa{r^2+2r-12w^2+12w-3}}{16w^3(1-w)^3}.
\end{eqnarray*}
(ii-3) If $r\in[2w-1,1-w]$, then
\begin{eqnarray*}
p^{(2)}_w(r) &=& \int^{\frac{1-r}w-1}_0\dif
r_1\int^{\frac{r+wr_1}{1-w}}_{\frac{r-wr_1}{1-w}}\dif
r_2+\int^{\frac rw}_{\frac{1-r}w-1}\dif
r_1\int^1_{\frac{r-wr_1}{1-w}}\dif r_2+\int^1_{\frac rw}\dif
r_1\int^1_{\frac{wr_1-r}{1-w}}\dif
r_2\\
&=&\frac{3(1-r)^2r\Pa{r^2+2r-12w^2+12w-3}}{16w^3(1-w)^3}.
\end{eqnarray*}
(ii-4) If $r\in\Br{0,2w-1}$, then
\begin{eqnarray*}
p^{(2)}_w(r) =\int^{\frac rw}_0\dif
r_1\int^{\frac{r+wr_1}{1-w}}_{\frac{r-wr_1}{1-w}}\dif r_2 +
\int^{\frac{1-r}w-1}_{\frac rw}\dif
r_1\int^{\frac{wr_1+r}{1-w}}_{\frac{wr_1-r}{1-w}}\dif r_2 +
\int^{\frac{1+r}w-1}_{\frac{1-r}w-1}\dif
r_1\int^1_{\frac{wr_1-r}{1-w}}\dif r_2=\frac{3r^2}{w^3}.
\end{eqnarray*}
From the above reasoning, we obtain the formula \eqref{eq:pdf2} for
$p^{(2)}_w(r)$ whenever $w\in\Pa{\frac12,1}$. This completes the
proof.
\end{proof}

To illustrate our methods, we choose the equiprobable mixture as a
toy model. Essentially, our methods applies to any probabilistic
mixture of qubits. Let $2\leqslant n\in\natural$. Denote by
\begin{eqnarray*}
\rho(\bss_n) = \frac1n\sum^n_{j=1}\rho(\bsr_j)\Longleftrightarrow
\bss_n = \frac1n\sum^n_{j=1}\bsr_j.
\end{eqnarray*}
And denote by $p^{(n)}(s_n)$ the distribution density of
$s_n=\abs{\bss_n}\in[0,1]$. Since
\begin{eqnarray*}
\rho(\bss_n) = \frac{n-1}n
\rho(\bss_{n-1})+\frac1n\rho(\bsr_n)\Longleftrightarrow
\bss_n=\frac{n-1}n\bss_{n-1}+\frac1n\bsr_n,
\end{eqnarray*}
it follows that
\begin{eqnarray}\label{eq:anyN}
p^{(n)}(s_n) = \iint_{R_{\frac{n-1}n}(s_n)}
f_{\frac{n-1}n}(s_n|s_{n-1},r_n)p^{(n-1)}(s_{n-1})p^{(1)}(r_n)\dif
s_{n-1}\dif r_n,
\end{eqnarray}
where $f_{\frac{n-1}n}(s_n|s_{n-1},r_n) =\frac{n^2}{2(n-1)}
\frac{s_n}{s_{n-1}r_n},\quad p^{(1)}(r_n)=3r^2_n$, and
\begin{eqnarray*}
R_{\frac{n-1}n}(s_n) = \Set{(s_{n-1},r_n)\in[0,1]^2:
\Abs{\frac{n-1}n s_{n-1}-\frac1n r_n}\leqslant s_n\leqslant
\frac{n-1}n s_{n-1}+\frac1n r_n}.
\end{eqnarray*}
When $n=2$, we give the detailed proof about the following elegant
result, the detailed proof (of Theorem~\ref{th:345}) for the larger
number $n>3$ can be found in Appendix~\ref{app:a}.

\begin{thrm}\label{th:2}
The probability density function of the length $r$ of the Bloch
vector $\bsr$ of the equiprobable mixture:
$\rho(\bsr)=\frac12\sum^2_{j=1}\rho_j$, where
$\rho_j\in\cE_{2,2}(j=1,2)$, is given by $p^{(2)}(r) =
12r^2(r^3-3r+2), r\in[0,1]$.
\end{thrm}

\begin{proof}
(1) If $r\in\Br{\frac12,1}$, then $R(r)=\Set{(r_1,r_2)\in[0,1]^2:
2r-1\leqslant r_1\leqslant 1, 2r-r_1\leqslant r_2\leqslant 1}$. Thus
$$
p^{(2)}(r) =18r\int^1_{2r-1}\dif r_1\Pa{r_1\int^1_{2r-r_1}r_2\dif
r_2}=12r^2(r^3-3r+2).
$$
(2) If $r\in\Br{\frac14,\frac12}$, then
\begin{eqnarray*}
R(r)&=&\Set{(r_1,r_2):0\leqslant r_1\leqslant 1-2r, 2r-r_1\leqslant
r_2\leqslant 2r+r_1}\\
&&\cup \Set{(r_1,r_2):1-2r\leqslant r_1\leqslant 2r, 2r-r_1\leqslant
r_2\leqslant 1}\\
&&\cup\Set{(r_1,r_2):2r\leqslant r_1\leqslant1, r_1-2r\leqslant
r_2\leqslant1}.
\end{eqnarray*}
Thus
\begin{eqnarray*}
p^{(2)}(r) &=& 18r\int^{1-2r}_0\dif
r_1\Pa{r_1\int^{2r+r_1}_{2r-r_1}r_2\dif r_2}+18r\int^{2r}_{1-2r}\dif
r_1\Pa{r_1\int^1_{2r-r_1}r_2\dif r_2}\\
&&+18r\int^1_{2r}\dif r_1\Pa{r_1\int^1_{r_1-2r}r_2\dif r_2}
=12r^2(r^3-3r+2).
\end{eqnarray*}
(3) If $r\in\Br{0,\frac14}$, then
\begin{eqnarray*}
R(r)&=&\Set{(r_1,r_2):0\leqslant r_1\leqslant 2r, 2r-r_1\leqslant
r_2\leqslant 2r+r_1}\\
&&\cup \Set{(r_1,r_2):2r\leqslant r_1\leqslant 1-2r, r_1-2r\leqslant
r_2\leqslant r_1+2r}\\
&&\cup\Set{(r_1,r_2):1-2r\leqslant r_1\leqslant1, r_1-2r\leqslant
r_2\leqslant1}.
\end{eqnarray*}
Thus
\begin{eqnarray*}
p^{(2)}(r) &=& 18r\int^{2r}_0\dif
r_1\Pa{r_1\int^{2r+r_1}_{2r-r_1}r_2\dif r_2}+18r\int^{1-2r}_{2r}\dif
r_1\Pa{r_1\int^{r_1+2r}_{r_1-2r}r_2\dif r_2}\\
&&+18r\int^1_{1-2r}\dif r_1\Pa{r_1\int^1_{r_1-2r}r_2\dif r_2}
=12r^2(r^3-3r+2).
\end{eqnarray*}
In summary, we conclude that $p^{(2)}(r) = 12r^2(r^3-3r+2),
r\in[0,1]$. We are done.
\end{proof}
In fact, for any $n$, we can use \eqref{eq:anyN} to derive the
density function $p^{(n)}(r)$ of $r=\abs{\bsr}$ for the equiprobable
mixture: $\rho(\bsr)=\frac1n\sum^n_{j=1}\rho_j$, where
$\rho_j\in\cE_{2,2}$ for $j=1,\ldots,n$. Indeed, for $n=3,4,5$, we
have the following result:
\begin{thrm}\label{th:345}
The probability density function $p^{(n)}(r),r\in[0,1]$, of the
length $r$ of the Bloch vector $\bsr$ of the equiprobable mixture:
$\rho(\bsr)=\frac1n\sum^n_{j=1}\rho_j$, where
$\rho_j\in\cE_{2,2}(j=1,\ldots,n)$, can be identified explicitly.
Specifically,
\begin{enumerate}[(i)]
\item for $n=3$, the density function is given by
\begin{eqnarray*}
p^{(3)}(r) =\begin{cases}f^{(3)}_R(r),& r\in\Br{\frac13,1},\\
f^{(3)}_L(r),&r\in\Br{0,\frac13}.
\end{cases}
\end{eqnarray*}
where $f^{(3)}_i(r)(i=L,R)$ are given in the following
\begin{eqnarray*}
f^{(3)}_R(r)&=&\frac{6561}{2240} (1-r)^4 r\Pa{9 r^3+36 r^2+27 r-2},\\
f^{(3)}_L(r)&=&-\frac{243}{1120} r^2 \Pa{243 r^6-1701 r^4+945
r^2-175}.
\end{eqnarray*}
\item for $n=4$, the density function is given by
\begin{eqnarray*}
p^{(4)}(r) = \begin{cases}f^{(4)}_R(r),& r\in\Br{\frac12,1},\\
f^{(4)}_L(r),&r\in\Br{0,\frac12}.
\end{cases}
\end{eqnarray*}
where $f^{(4)}_i(r)(i=L,R)$ are given in the following
\begin{eqnarray*}
f^{(4)}_R(r) &=& \frac{512}{175}(1 - r)^6 r\Pa{16 r^4 + 96 r^3 + 156 r^2 + 56 r  -9},\\
f^{(4)}_L(r) &=& - \frac{128}{175}r^2 \Pa{192 r^9 - 2160 r^7 + 960
r^6 + 3780 r^5 - 3528 r^4 + 720 r^2 -85}.
\end{eqnarray*}
\item for $n=5$, the density function is given by
\begin{eqnarray*}
p^{(5)}(r) = \begin{cases}f^{(5)}_R(r),& r\in\Br{\frac35,1},\\
f^{(5)}_M(r),&r\in\Br{\frac15,\frac35},\\
f^{(5)}_L(r),&r\in\Br{0,\frac15}.
\end{cases}
\end{eqnarray*}
where $f^{(5)}_i(r)(i=L,M,R)$ are given in the following
{\scriptsize
\begin{eqnarray*}
f^{(5)}_R(r) &=& \frac{1953125}{16400384} (1-r)^8 r \Pa{625 r^5 +
5000 r^4 + 12750 r^3 + 11300 r^2 + 1825 r - 612},\\
f^{(5)}_M(r) &=& -\frac5{4100096} r \left(244140625 r^{13}
-3808593750 r^{11} + 2792968750 r^{10}+ 12568359375 r^9 -
19103906250 r^8 - 670312500 r^7 \right.\\
&&~~~~~~~~~~~~~~~~~~\left.+ 18098437500 r^6 - 12978590625 r^5+ 2511437500 r^4 + 360038250 r^3 + 44625750 r^2 - 75822175 r + 67086\right),\\
f^{(5)}_L(r) &=& \frac{125}{8200192} r^2 \Pa{29296875 r^{12} -
457031250 r^{10} + 1508203125 r^8- 938437500 r^6 + 316441125 r^4 -
63050130 r^2 + 5855707}.
\end{eqnarray*}}
\end{enumerate}
\end{thrm}


The density curves of the mixtures for $n=2,3,4,5$, mentioned in
Theorem~\ref{th:2} and Theorem~\ref{th:345}, are plotted in the same
coordinate system, see Figure~\ref{fig:p2345}.
\begin{figure}[ht]\centering
{\begin{minipage}[b]{0.7\linewidth}
\includegraphics[width=1\textwidth]{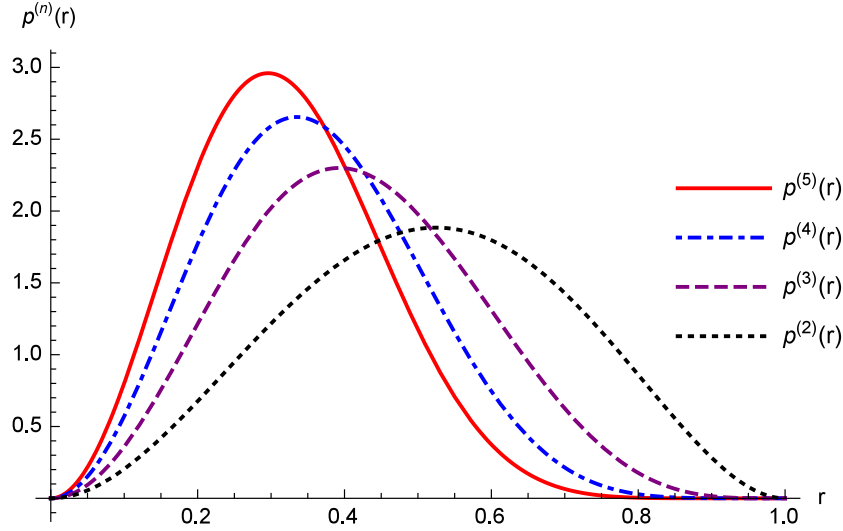}
\end{minipage}}
\caption{The density curves for the number $n=2,3,4,5$.}
\label{fig:p2345}
\end{figure}
We see from Figure~\ref{fig:p2345} that the points at which the peak
values are attained are moved closer to $y$-axis from right to left,
at the same time, the peak values become larger and larger. This
implies that the mixture of qubits gradually approaches the
completely mixed state when the component number $n$ increases in
the mixture.

Next, we consider to derive the density of a diagonal entry of the
mixture of two qubit states. In fact,
\begin{eqnarray*}
    &&\dh^T_{\cO_\mu\times\cO_\nu} = \dh^T_{\cO_\mu}*\dh^T_{\cO_\nu}\\
    &&= \Pa{\delta_{\mu\cdot\bsh} -
        \delta_{-\mu\cdot\bsh}}*\Pa{\delta_{\nu\cdot\bsh} -
        \delta_{-\nu\cdot\bsh}}* H_{-\bsh}* H_{-\bsh}\\
    &&=\Pa{\delta_{(\mu+\nu)\cdot\bsh} + \delta_{-(\mu+\nu)\cdot\bsh} -
        \delta_{(\mu-\nu)\cdot\bsh} - \delta_{-(\mu-\nu)\cdot\bsh}}*
    H_{-\bsh}* H_{-\bsh}.
\end{eqnarray*}
Based on this, we see that, $t\cdot\bsh\in\mathrm{i}\liet^*_{>0}$,
i.e., for $t>0$,
\begin{eqnarray*}
    \Inner{\dh^T_{\cO_\mu\times\cO_\nu}}{f} = \int^{\mu+\nu}_0 f(t\cdot\bsh)(\mu+\nu-t)\dif t -
    \int^{\abs{\mu-\nu}}_0f(t\cdot\bsh)(\abs{\mu-\nu}-t)\dif t.
\end{eqnarray*}
The density of $t\cdot\bsh\in\mathrm{i}\liet^*_{>0}$ with respect to
the measure $\dh^T_{\cO_\mu\times\cO_\nu}$ is as follows.
\begin{eqnarray}\label{eq:diag2}
\begin{cases}\mu+\nu - \abs{\mu-\nu},&
t\in\Br{0,\abs{\mu-\nu}};\\ \mu+\nu - t, &
t\in\Br{\abs{\mu-\nu},\mu+\nu}.\end{cases}
\end{eqnarray}

\begin{prop}
The probability density function of a diagonal entry of the mixture
$\rho^w_s=w\rho_1+(1-w)\rho_2 (w\in(0,1))$ of two random density
matrices, chosen uniformly from respective unitary orbits $\cO_a$
and $\cO_b$ with $a,b$ are fixed in $\Pa{0,\frac12}$ is given by
\begin{eqnarray}\label{eq:diagw}
q_w(x|a,b) = \frac1{4w(1-w)\Pa{\frac12-a}\Pa{\frac12-b}}
\begin{cases}x - t_0(w),& x\in\Br{t_0(w),t_1(w)};\\
t_1(w)-t_0(w), & x\in\Br{t_1(w),1-t_1(w)};\\
-x+(1-t_0(w)), & x\in\Br{1-t_1(w),1-t_0(w)}.
\end{cases}
\end{eqnarray}
Here the notations $t_0(w)$ and $t_1(w)$ can be found in
\eqref{eq:t0-t1-of-w} in Proposition~\ref{prop:lmn}. In particular,
for $w=\frac12$, we get that
\begin{eqnarray}\label{eq:diag0.5}
q(x|a,b) = \frac1{\Pa{\frac12-a}\Pa{\frac12-b}}
\begin{cases}x - t_0,& x\in\Br{t_0,t_1};\\
t_1-t_0, & x\in\Br{t_1,1-t_1};\\
-x+(1-t_0), & x\in\Br{1-t_1,1-t_0}.
\end{cases}
\end{eqnarray}
Here $t_0=\frac{a+b}2$ and $t_1=\frac{1-\abs{a-b}}2$ for given $a$
and $b$.
\end{prop}

\begin{proof}

Now for the mixture $\rho^w_s = w\rho_1+(1-w)\rho_2(w\in(0,1))$
where $\rho_1\in\cO_a$ and $\rho_2\in\cO_b$, let $\rho^{w,D}_s
=\diag(x,1-x)$, where $x\in\Br{0,\frac12}$. Then for
$a,b\in\Pa{0,\frac12}$, assume that
$\mu=w\Pa{\frac12-a},\nu=(1-w)\Pa{\frac12-b}, t=\frac12-x$. By
substituting these parameters into the above expression
\eqref{eq:diag2} and relaxing the constraint $x\in\Br{0,\frac12}$ to
$x\in\Br{0,1}$ at the same time keeping $a,b$ fixed in
$\Pa{0,\frac12}$, then after normalizing it, we get that the desired
identity \eqref{eq:diagw}. This result can also be derived from the
method used in \cite{Zhang2017c}. Now for the weight $w=\frac12$, by
substituting these parameters into the above expression, we get
Eq.~\eqref{eq:diag0.5}. It can also be written as
\begin{eqnarray*}
q(x|a,b) = \frac1{2\Pa{\frac12-a}\Pa{\frac12-b}}
(\abs{x-t_0}+\abs{x-(1-t_0)} - \abs{x-t_1}-\abs{x-(1-t_1)}).
\end{eqnarray*}
We have done it.
\end{proof}
\subsection{The mixture of three qubit states}
Let $\lambda\in\mathrm{i}\liet^*_{>0}$, and $\mu,\nu\in
\mathrm{i}\liet^*_{\geqslant0}$. For the mixture of three qubit
states, we have that the Abelian D-H measure over the manifold
$\cO_\lambda\times\cO_\mu\times\cO_\nu$ is given by the following
convolution: $\dh^T_{\cO_\lambda\times\cO_\mu\times\cO_\nu} =
\dh^T_{\cO_\lambda}*\dh^T_{\cO_\mu}*\dh^T_{\cO_\nu}$. Thus the
non-Abelian D-H measure is the following:
$\dh^K_{\cO_\lambda\times\cO_\mu\times\cO_\nu} = \sum_{w\in
W}(-1)^{l(w)} \delta_{w\lambda}*\dh^T_{\cO_\mu}*\dh^T_{\cO_\nu}$. It
follows that
\begin{eqnarray*}
&&\dh^K_{\cO_\lambda\times\cO_\mu\times\cO_\nu} = \sum_{w\in
W}(-1)^{l(w)} \delta_{w\lambda}*\dh^T_{\cO_\mu}*\dh^T_{\cO_\nu}\\
&&=\left. \Pa{\sum_{w\in S_2}(-1)^{l(w)}
\delta_{w\lambda}}*\Pa{\sum_{w\in
S_2}(-1)^{l(w)}\delta_{w\mu}}*\Pa{\sum_{w\in
S_2}(-1)^{l(w)}\delta_{w\nu}}* H_{-\alpha}* H_{-\alpha}\right|_{\mathrm{i}\liet^*_{\geqslant0}}\\
&&= \left.(\delta_{\lambda\cdot \bsh} - \delta_{-\lambda\cdot
\bsh})*(\delta_{\mu\cdot \bsh} - \delta_{-\mu\cdot
\bsh})*(\delta_{\nu\cdot
\bsh} - \delta_{-\nu\cdot \bsh})* H_{-\bsh}* H_{-\bsh}\right|_{\mathrm{i}\liet^*_{\geqslant0}}\\
&&= \left.\Pa{\delta_{(\lambda+\mu+\nu)\cdot \bsh} +
\delta_{-(\lambda+\mu-\nu)\cdot \bsh}
+\delta_{(-\lambda+\mu-\nu)\cdot
\bsh}+\delta_{(\lambda-\mu-\nu)\cdot \bsh} }* H_{-\bsh}*
H_{-\bsh}\right|_{\mathrm{i}\liet^*_{\geqslant0}}\\
&&~~~-\left.\Pa{\delta_{-(\lambda+\mu+\nu)\cdot \bsh}  +
\delta_{(\lambda+\mu-\nu)\cdot \bsh} +
\delta_{(\lambda-\mu+\nu)\cdot \bsh}+\delta_{(-\lambda+\mu+\nu)\cdot
\bsh}}* H_{-\bsh}*
H_{-\bsh}\right|_{\mathrm{i}\liet^*_{\geqslant0}}.
\end{eqnarray*}
Without loss of generality, we assume that
$\lambda\geqslant\mu\geqslant\nu\geqslant0$ by the permutation
symmetry of the triple $(\lambda,\mu,\nu)$. Denote
$a_3:=\lambda+\mu+\nu,\quad a_2:=\lambda+\mu-\nu,\quad
a_1:=\lambda-\mu+\nu$. Under the above assumption, i.e.,
$\lambda\geqslant\mu\geqslant\nu\geqslant0$, we cannot identify the
sign of $-a_0:=-\lambda+\mu+\nu$ or $a_0:=\lambda-\mu-\nu$. In fact,
$a_3\geqslant a_2\geqslant a_1\geqslant\abs{a_0}\geqslant0$.
Therefore, $\dh^K_{\cO_\lambda\times\cO_\mu\times\cO_\nu} =
\Pa{\delta_{a_3\cdot \bsh}+\sign(a_0)\delta_{\abs{a_0}\cdot
\bsh}-\delta_{a_2\cdot \bsh} - \delta_{a_1\cdot \bsh}}* H_{-\bsh}*
H_{-\bsh}$, where $\sign(a_0)$ is the sign function that extracts
the sign of a real number. Furthermore, we have
\begin{eqnarray*}
\Inner{\dh^K_{\cO_\lambda\times\cO_\mu\times\cO_\nu}}{f}
&=&\int^{a_3}_0f[(a_3-t)\cdot \bsh]t\dif t
+\sign(a_0)\int^{\abs{a_0}}_0f[(\abs{a_0}-t)\cdot
\bsh]t\dif t\\
&&-\int^{a_2}_0f[(a_2-t)\cdot \bsh]t\dif t -\int^{a_1}_0
f[(a_1-t)\cdot \bsh]t\dif t.
\end{eqnarray*}
Thus we have
\begin{eqnarray*}
\Inner{\dh^K_{\cO_\lambda\times\cO_\mu\times\cO_\nu}}{f}
&=&\int^{a_3}_0f(t\cdot \bsh)(a_3-t)\dif t
+\int^{\abs{a_0}}_0f(t\cdot
\bsh)\sign(a_0)(\abs{a_0}-t)\dif t\\
&&-\int^{a_2}_0f(t\cdot \bsh)(a_2-t)\dif t -\int^{a_1}_0 f(t\cdot
\bsh)(a_1-t)\dif t.
\end{eqnarray*}
According to the definition of non-Abelian D-H measure, we see that
\begin{eqnarray*}
\dh^K_M = \frac1{p_K}(\tau_K)_*(\Phi_K)_*\mu_M,
\end{eqnarray*}
where $K=\rS\rU(2)$ and $M=\cO_\lambda\otimes\cO_\mu\otimes\cO_\nu$.
Multipling the non-Abelian Duistermaat-Heckman measure by the
symplectic volume polynomial $p_K(\zeta\cdot\bsh)=2\zeta$, thus we
see that
\begin{eqnarray*}
\bP_{\mathrm{eig}}:=p_K\frac{\dh^K_M}{\vol(M)} =
(\tau_K)_*(\Phi_K)_*\Pa{\frac{\mu_M}{\vol(M)}},
\end{eqnarray*}
where $\vol(M) =
\vol(\cO_\lambda)\vol(\cO_\mu)\vol(\cO_\nu)=8\lambda\mu\nu$. Next,
\begin{eqnarray*}
&&\Inner{\bP_{\mathrm{eig}}}{f} =
\Inner{(\tau_K)_*(\Phi_K)_*\Pa{\frac{\mu_M}{\vol(M)}}}{f} =
\Inner{p_K\frac{\dh^K_M}{\vol(M)}}{f}\\
&&=\Inner{\frac{p_K}{8\lambda\mu\nu}\dh^K_M}{f}=\frac1{8\lambda\mu\nu}\Inner{\dh^K_M}{p_Kf}.
\end{eqnarray*}
Then for $\zeta\in\mathrm{i}\liet^*_{\geqslant0}$, we have
\begin{eqnarray*}
\Inner{\dh^K_M}{p_Kf} &=&2\int^{a_3}_0 f(\zeta\cdot
\bsh)\zeta(a_3-\zeta)\dif \zeta +2\int^{\abs{a_0}}_0 f(\zeta\cdot
\bsh)\sign(a_0)\zeta(\abs{a_0}-\zeta)\dif \zeta\\
&&-2\int^{a_2}_0 f(\zeta\cdot \bsh)\zeta(a_2-\zeta)\dif \zeta
-2\int^{a_1}_0
 f(\zeta\cdot \bsh)\zeta(a_1-\zeta)\dif \zeta.
\end{eqnarray*}
Therefore
\begin{eqnarray*}
4\lambda\mu\nu\Inner{\bP_{\mathrm{eig}}}{f}&=&\int^{a_3}_0
f(\zeta\cdot \bsh)\zeta(a_3-\zeta)\dif \zeta +\int^{\abs{a_0}}_0
f(\zeta\cdot
\bsh)\sign(a_0)\zeta(\abs{a_0}-\zeta)\dif \zeta\\
&&-\int^{a_2}_0 f(\zeta\cdot \bsh)\zeta(a_2-\zeta)\dif \zeta
-\int^{a_1}_0
 f(\zeta\cdot \bsh)\zeta(a_1-\zeta)\dif \zeta.
\end{eqnarray*}
From this, we see that the density of $\zeta$ is given by the
following analytical formula:
\begin{eqnarray}\label{eq:3qubit-mix}
\dif\bP_{\mathrm{eig}}/\dif\zeta=p(\zeta|\lambda,\mu,\nu)
=\frac{\zeta}{4\lambda\mu\nu}
\begin{cases}
(1-\sign(a_0))\zeta,&\text{if }\zeta\in[0,\abs{a_0}],\\
\zeta+a_3-a_1-a_2,&\text{if }\zeta\in[\abs{a_0},a_1],\\
a_3-a_2,&\text{if }\zeta\in[a_1,a_2],\\
a_3-\zeta,&\text{if }\zeta\in[a_2,a_3].
\end{cases}
\end{eqnarray}

\begin{thrm}\label{prop:3Qbit}
The probability density function $p(s|a,b,c)$ of the minimal
eigenvalue $s$ of the equiprobable mixture of three random density
matrices, chosen uniformly from respective unitary orbits
$\cO_a:=\cO_{(a,1-a)},\cO_b:=\cO_{(b,1-b)}$, and
$\cO_c:=\cO_{(c,1-c)}$ where $a,b,c$ are fixed in $\Pa{0,\frac12}$
with $a\geqslant b\geqslant c$, is given by
\begin{eqnarray}\label{eq:3qeigen}
p(s|a,b,c)
=\frac{27}{4}\frac{\frac12-s}{\Pa{\frac12-a}\Pa{\frac12-b}\Pa{\frac12-c}}\Phi_{a,b,c}(s),
\end{eqnarray}
where
$$
\Phi_{a,b,c}(s) =
\begin{cases}\Pa{1+\sign\Pa{\frac16 -
\frac{a+b-c}3}}\Pa{\frac12-s},&\text{if } s\in\Br{T_3,
\frac12},\\
-s+\frac23-\frac{a+b-c}3,&\text{if } s\in\Br{T_2,
T_3},\\
T_1-T_0,&\text{if } s\in\Br{T_1,
T_2},\\
s-T_0,&\text{if } s\in\Br{T_0, T_1}.
\end{cases}
$$
Here
\begin{eqnarray}\label{eq:Tj}
\begin{cases}
T_0=\frac{a+b+c}3,\\
T_1=\frac13-\frac{a-b-c}3,\\
T_2=\frac13-\frac{-a+b-c}3,\\
T_3=\frac12-\abs{\frac16-\frac{a+b-c}3}.
\end{cases}
\end{eqnarray}
\end{thrm}

\begin{proof}
In fact, let $\rho_x\in\cO_x$ for $x\in\set{a,b,c}$, consider the
equiprobable mixture of $\rho_a,\rho_b,\rho_c$, i.e., $\rho_s =
\frac13(\rho_a+\rho_b+\rho_c)$ with $a,b,c,s\in\Br{0,\frac12}$.
Assume that $a\geqslant b\geqslant c$. Let
$$
\lambda:=\frac13\Pa{\frac12-c},\mu:=\frac13\Pa{\frac12-b},
\nu:=\frac13\Pa{\frac12-a}, \zeta:=\frac12-s.
$$
Then $\lambda\geqslant\mu\geqslant\nu\geqslant0$ and
$$
a_3=\frac12 - \frac{a+b+c}3, \quad a_2=\frac16+\frac{a-b-c}3,\quad
a_1=\frac16+\frac{-a+b-c}3, \quad a_0 = -\frac16 + \frac{a+b-c}3.
$$
That is,
\begin{eqnarray*}
&&\frac12 - a_3=\frac{a+b+c}3:=T_0, \quad \frac12 -
a_2=\frac13-\frac{a-b-c}3:=T_1,\\
&&\frac12 - a_1=\frac13-\frac{-a+b-c}3:=T_2, \quad \frac12 -
\abs{a_0} = \frac12 - \abs{\frac16 - \frac{a+b-c}3}:=T_3.
\end{eqnarray*}
Substituting these new symbols into \eqref{eq:3qubit-mix} directly
gives the desired result, i.e., Eq.~\eqref{eq:3qeigen}.
\end{proof}

\begin{remark}
If we do not require $s$ being the minimal eigenvalue, then we get
that the probability density function of an eigenvalue $s$ of the
equiprobable mixture of three random density matrices, chosen
uniformly from respective unitary orbits $\cO_a,\cO_b$, and $\cO_c$
where $a,b,c$ are fixed in $\Pa{0,\frac12}$ with $a\geqslant
b\geqslant c$, is given
\begin{eqnarray*}
p(s|a,b,c)
=\frac{27}{8}\frac{\frac12-s}{\Pa{\frac12-a}\Pa{\frac12-b}\Pa{\frac12-c}}\widetilde\Phi_{a,b,c}(s),
\end{eqnarray*}
where
$$
\widetilde\Phi_{a,b,c}(s) =
\begin{cases}
s-(1-T_0), &\text{if }s\in\Br{1-T_1,1-T_0},\\
-(T_1-T_0), &\text{if }s\in\Br{1-T_2,1-T_1},\\
-s+\frac13+\frac{a+b-c}3,&\text{if }s\in\Br{1-T_3,1-T_2},\\
\Pa{1+\sign\Pa{\frac16 -
\frac{a+b-c}3}}\Pa{\frac12-s},&\text{if } s\in\Br{T_3,
1-T_3},\\
-s+\frac23-\frac{a+b-c}3,&\text{if } s\in\Br{T_2,
T_3},\\
T_1-T_0,&\text{if } s\in\Br{T_1,
T_2},\\
s-T_0,&\text{if } s\in\Br{T_0,
T_1}.\\
\end{cases}
$$
In fact, we obtain the following result by the method in
\cite{Zhang2017c}:
\begin{eqnarray*}
\widetilde\Phi_{a,b,c}(s) &=& (\abs{s-T_0} - \abs{s-(1-T_0)}) -(\abs{s-T_1} - \abs{s-(1-T_1)})\\
&& -(\abs{s-T_2} - \abs{s-(1-T_2)}) -(\abs{s-T_3} -
\abs{s-(1-T_3)}).
\end{eqnarray*}
\end{remark}

\begin{cor}
The conditional probability density function of the length $r$ of
Bloch vector of the equiprobable mixture:
$\rho(\bsr)=\frac{\rho(\bsr_1)+\rho(\bsr_2)+\rho(\bsr_3)}{3}$, where
$r_1=\abs{\bsr_1},r_2=\abs{\bsr_2},r_3=\abs{\bsr_3}\in[0,1]$ are
fixed with $r_1\leqslant r_2\leqslant r_3$ and
$r_1+r_2-r_3\geqslant0$, is given by
\begin{eqnarray}\label{eq:p(r)}
p(r|r_1,r_2,r_3)=\frac94\frac1{r_1r_2r_3}
\begin{cases}
-3r^2+(r_1+r_2+r_3)r,&r\in\Br{\frac{-r_1+r_2+r_3}{3},\frac{r_1+r_2+r_3}{3}},\\
2r_1r,&r\in\Br{\frac{r_1-r_2+r_3}{3},\frac{-r_1+r_2+r_3}{3}},\\
3r^2+(r_1+r_2-r_3)r,&r\in\Br{\frac{r_1+r_2-r_3}{3},\frac{r_1-r_2+r_3}{3}},\\
r^2,&r\in\Br{0,\frac{r_1+r_2-r_3}{3}}.
\end{cases}
\end{eqnarray}
\end{cor}

\begin{proof}
We have already known that two eigenvalues of a qubit density matrix
can be identified by the length of its Bloch vector:
$\lambda_\pm(\rho)=\frac12\Pa{1\pm r}$, where $r\in[0,1]$. Now let
$(r_1,r_2,r_3)=(1-2a,1-2b,1-2c)$ and $r=1-2s$, where $a,b,c,s$ are
from Proposition~\ref{prop:3Qbit}. Using these new symbols
$(r_1,r_2,r_3,r)$ instead of $(a,b,c,s)$, some computations give the
desired result.
\end{proof}

Similarly, we can derive the density of diagonal part $\rho^D_s$ as
follows:
\begin{eqnarray}\label{eq:q(x|a,b,c)}
q(x|a,b,c)=
\frac{27}{16}\frac1{(\frac1{2}-a)(\frac1{2}-b)(\frac1{2}-c)}\Psi_{a,b,c}(x),
\end{eqnarray}
where
\begin{eqnarray*}
\Psi_{a,b,c}(x):=\begin{cases}
(x-T_0)^2,&x\in[T_0,T_1],\\
(T_1-T_0)(2x-T_0-T_1),&x\in[T_1,T_2],\\
-x^2+T_0^2-T_1^2-T_2^2+2x(1-T_3),&x\in[T_2,T_3],\\
-2x^2+2x+T_0^2-T_1^2-T_2^2-T_3^2,&x\in[T_3,1-T_3],\\
-(x-T_3)^2+(T_3-1)^2+T_0^2-T_1^2-T_2^2,&x\in[1-T_3,1-T_2],\\
(T_1-T_0)(2-2x-T_0-T_1),&x\in[1-T_2,1-T_1],\\
(x-(1-T_0))^2,&x\in[1-T_1,1-T_0].
\end{cases}
\end{eqnarray*}
Here $T_j,j=0,1,2,3$ are from \eqref{eq:Tj}.

As a demonstration, we plot the density function of diagonal part
$\rho^D_s$ and the spectral density function of an eigenvalue of
$\rho_s$, where $\rho_s=\frac13(\rho_1+\rho_2+\rho_3)$, in the qubit
situation, see Figure~\ref{fig:qx-vs-ps}.
\begin{figure}[ht]\centering
\subfigure[$q(x|a,b,c)$ vs $x$] {\begin{minipage}[b]{0.47\linewidth}
\includegraphics[width=1\textwidth]{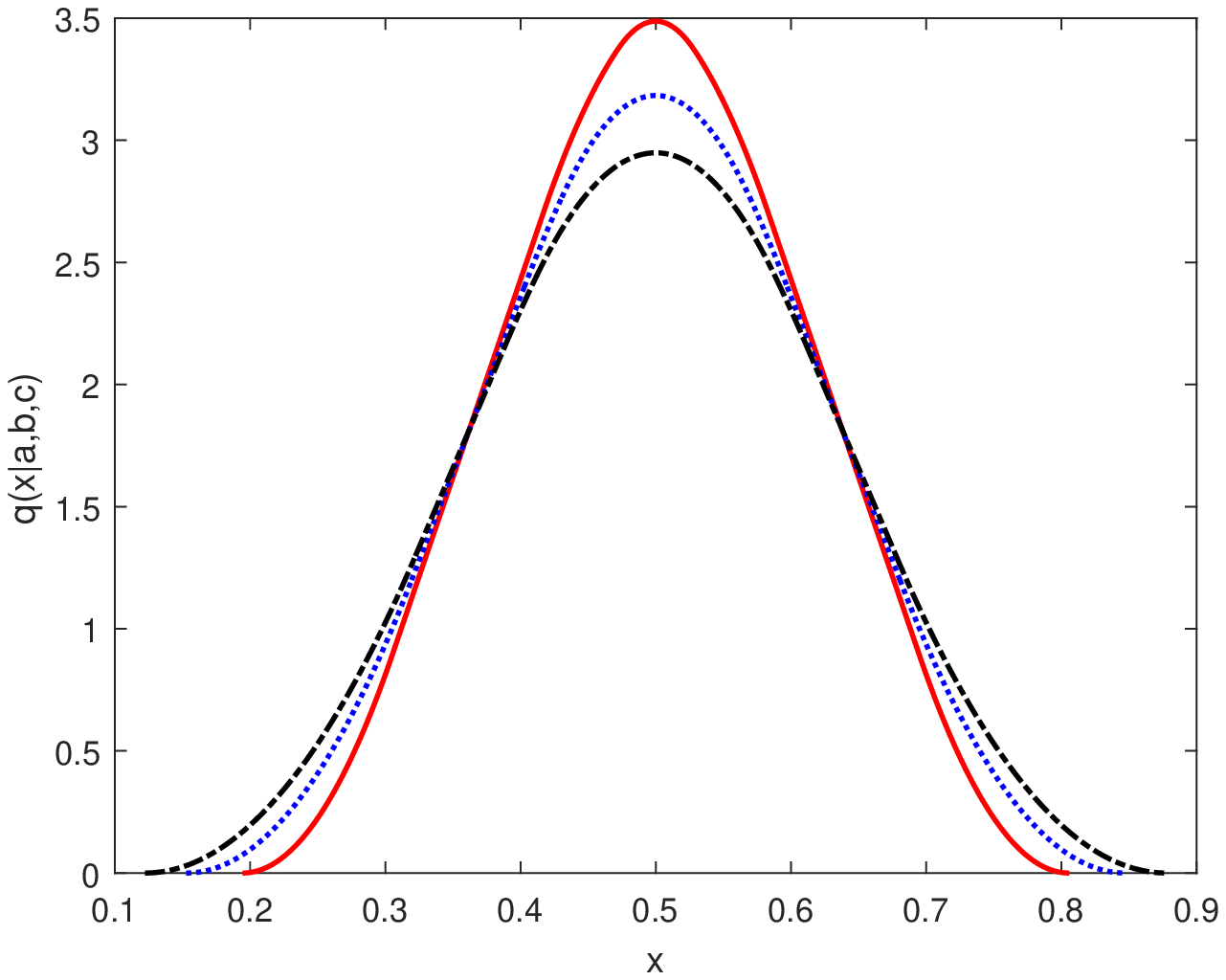}
\end{minipage}}
\subfigure[$p(s|a,b,c)$ vs $s$] {\begin{minipage}[b]{0.47\linewidth}
\includegraphics[width=1\textwidth]{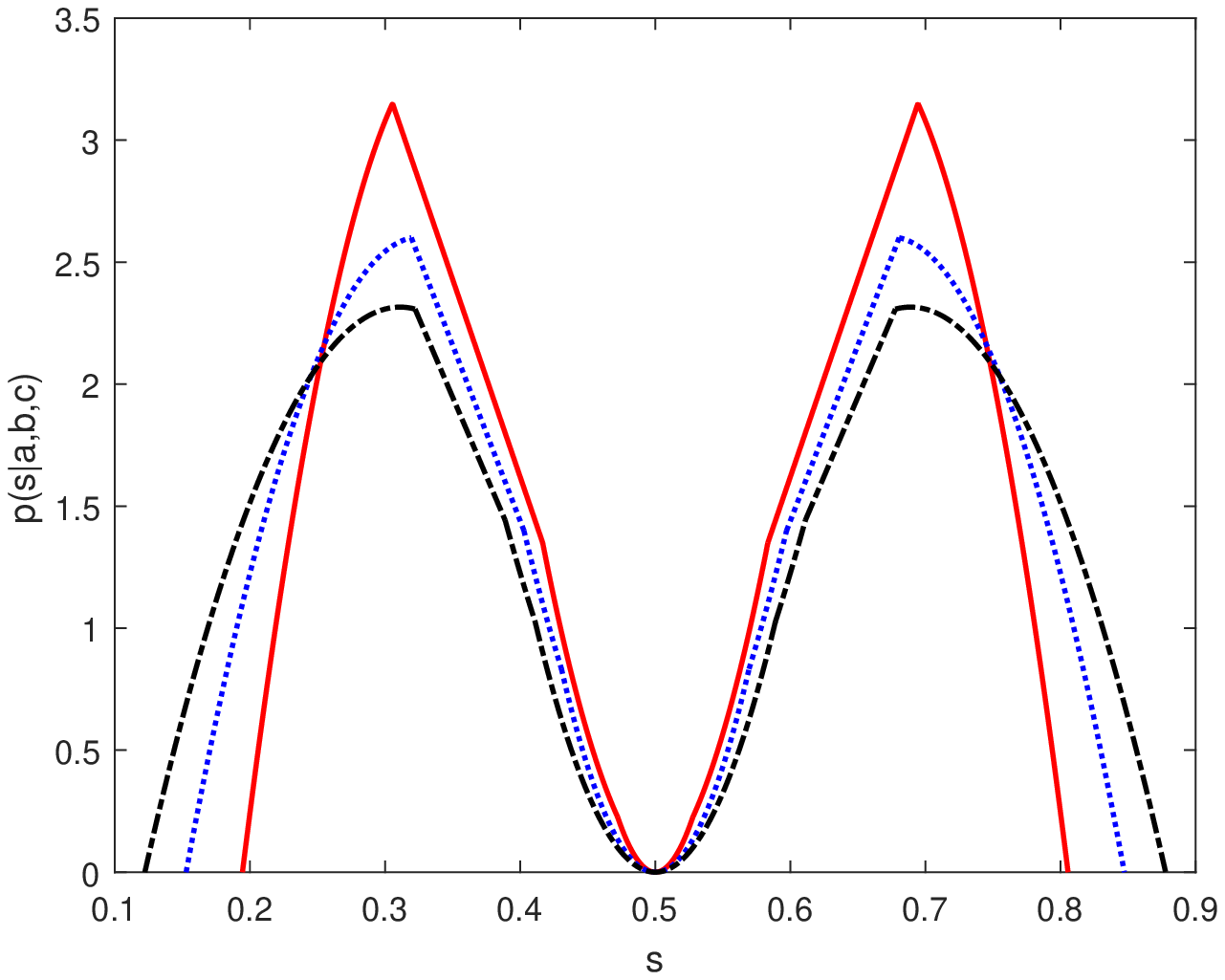}
\end{minipage}}
\caption{(a) The density function of the diagonal part $\rho^D_s$:
$q(x|a,b,c)$ versus $x$. (b) The density function of an eigenvalue
of $\rho_s$: $p(s|a,b,c)$ versus $s$. In either two figures, they
are demonstrated with red line corresponding to
$(a,b,c)=\Pa{\frac13,\frac16,\frac1{12}}$, blue line corresponding
to $(a,b,c)=\Pa{\frac14,\frac18,\frac1{12}}$, black line
corresponding to $(a,b,c)=\Pa{\frac15,\frac1{10},\frac1{15}}$.}
\label{fig:qx-vs-ps}
\end{figure}
From Figure~\ref{fig:qx-vs-ps}, we see that the graphs of their
distribution densities are symmetric with respect to the vertical
line $\frac12$ in the coordinate system. We find that for the
mixture of three random density matrices, the density of a generic
eigenvalue taking $\frac12$ is vanished (see
Fig~\ref{fig:qx-vs-ps}(a)), but the density of a generic diagonal
entry taking $\frac12$ is largest (see Fig~\ref{fig:qx-vs-ps}(b)).
Note that
$\Pa{\frac15,\frac1{10},\frac1{15}}\prec\Pa{\frac14,\frac18,\frac1{12}}\prec\Pa{\frac13,\frac16,\frac1{12}}$
\footnote{Here $\prec$ means the majorization. That is, for two
$d$-dimensional real vectors $\bsp=(p_1,\ldots,p_d)$ and
$\bsq=(q_1,\ldots,q_d)$, we say that $\bsp$ is majorized by $\bsq$,
denoted by $\bsp\prec\bsq$, if $\sum^k_{j=1}p^\downarrow_j\leqslant
\sum^k_{j=1}q^\downarrow_j$ for all $k\in\set{1,\ldots,d-1}$, where
$v^\downarrow$ represents the vector $v$ with entries arranged in
non-increasing order.}. The common feature in the
Figure~\ref{fig:qx-vs-ps}(a\&b) is the graphs rising up largely in
the sense of the majorization order, i.e., the top of the graph
corresponding to $\bsp$ is lower than that of the graph
corresponding to $\bsq$ if $\bsp\prec\bsq$.

\begin{figure}[ht]\centering
{\begin{minipage}[b]{0.7\linewidth}
\includegraphics[width=1\textwidth]{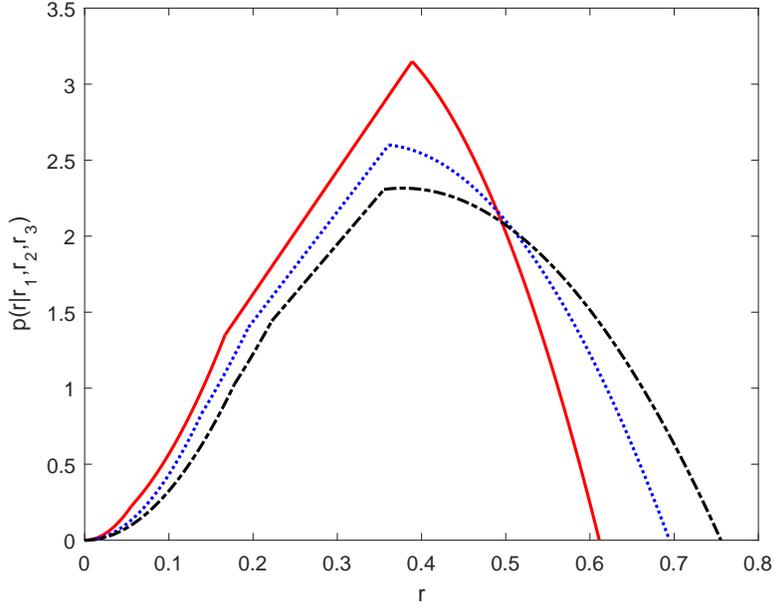}
\end{minipage}}
\caption{Reformulation of $p(s|a,b,c)$ in terms of Bloch lengths via
$a=\frac{1-r_1}2,b=\frac{1-r_2}2,c=\frac{1-r_3}2$ and $s=\frac{1\pm
r}2$: $p(r|r_1,r_2,r_3)$ versus $r$ with red line corresponding to
$(r_1,r_2,r_3)=\Pa{\frac13,\frac23,\frac56}$, blue line
corresponding to $(r_1,r_2,r_3)=\Pa{\frac12,\frac34,\frac56}$, black
line corresponding to
$(r_1,r_2,r_3)=\Pa{\frac35,\frac45,\frac{13}{15}}$.} \label{fig:pr}
\end{figure}
\subsection{The mixture of two qutrit states}

Given $K=\rS\rU(3)$. Then
$T=\Set{\diag(e^{\mathrm{i}\theta_1},e^{\mathrm{i}\theta_2},e^{\mathrm{i}\theta_3}):\theta_1,\theta_2,\theta_3\in\real\bigwedge\sum^3_{j=1}\theta_j=0}$
with its Lie algebra
$\liet=\Set{\diag(\mathrm{i}\theta_1,\mathrm{i}\theta_2,\mathrm{i}\theta_3):\theta_1,\theta_2,\theta_3\in\real\bigwedge\sum^3_{j=1}\theta_j=0}$.
Thus $\bsh_1=\diag(1,-1,0),\bsh_2=\diag(0,1,-1)$ is the basis of
$\liet$, and $\liet= \mathrm{i}\real\cdot \bsh_1\oplus
\mathrm{i}\real\cdot \bsh_2\cong\real^2$. All positive roots
$\alpha_{ij}\in\liet^*$ with $i<j$, where $i,j\in\set{1,2,3}$, are
given by $\alpha_{12}\cong\bsh_1,\alpha_{13}\cong\bsh_1+\bsh_2,
\alpha_{23}\cong\bsh_2$. Note that the above realizations are via
the Hilbert-Schmidt inner product. Furthermore, we can construct an
orthonormal basis for $\mathrm{i}\liet$ as
$\bsh_x=\frac{\bsh_1+\bsh_2}{\sqrt{2}},\bsh_y=\frac{\bsh_1-\bsh_2}{\sqrt{6}}$.
Throughout this section, denote $u:=\sqrt{\frac12}$ and
$v:=\sqrt{\frac32}$. Then all positive roots $\alpha_{ij}\in\liet^*$
with $i<j$, where $i,j\in\set{1,2,3}$, via $(\bsh_x,\bsh_y)$, are
given by $\alpha_{12}\cong u\bsh_x+v\bsh_y,\alpha_{13}\cong
2u\bsh_x, \alpha_{23}\cong u\bsh_x-v\bsh_y$. Besides, the Weyl group
is given by $W=S_3$. Let $\lambda\in\mathrm{i}\liet^*_{>0}$ and
$\mu\in\mathrm{i}\liet^*_{\geqslant0}$. We see that
\eqref{eq:NA-DH-measure} reduces to the following:
$\dh^K_{\cO_\lambda\times\cO_\mu} = \Pa{\sum_{w\in
S_3}(-1)^{l(w)}\delta_{w\lambda}}*\dh^T_{\cO_\mu}$, where
$\dh^T_{\cO_\mu} = \sum_{w\in S_3}(-1)^{l(w)}\delta_{w\mu}*
H_{-\alpha_{12}}*H_{-\alpha_{13}}*H_{-\alpha_{23}}$. So,
\begin{eqnarray}\label{eq:sum-qutrit}
\dh^K_{\cO_\lambda\times\cO_\mu} = \Pa{\sum_{w\in
S_3}(-1)^{l(w)}\delta_{w\lambda}}*\Pa{\sum_{w\in
S_3}(-1)^{l(w)}\delta_{w\mu}}*
H_{-\alpha_{12}}*H_{-\alpha_{13}}*H_{-\alpha_{23}}\big|_{\mathrm{i}\liet^*_{\geqslant0}}.
\end{eqnarray}
We see that
$\lambda=(\lambda_1,\lambda_2,\lambda_3)\in\mathrm{i}\liet^*_{>0}$,
i.e., $\lambda_1>\lambda_2>\lambda_3$ and $\sum^3_{j=1}\lambda_j=0$.
Note that $\lambda_1>0>\lambda_3$. With the orthonormal basis
$\set{\bsh_x,\bsh_y}$ of $\mathrm{i}\liet$, $\lambda\cong
u(\lambda_1-\lambda_3)\bsh_x + v(\lambda_1+\lambda_3)\bsh_y\cong
(u(\lambda_1-\lambda_3),v(\lambda_1+\lambda_3))$. Thus
\begin{eqnarray*}
\Pa{\sum_{w\in S_3}(-1)^{l(w)}\delta_{w\lambda}}*\Pa{\sum_{w\in
S_3}(-1)^{l(w)}\delta_{w\mu}}= \sum_{w,w'\in
S_3}(-1)^{l(w)+l(w')}\delta_{w\lambda+w'\mu}.
\end{eqnarray*}
Specifically, we see that, via $\set{\bsh_x,\bsh_y}$,
\begin{eqnarray*}
(\lambda_1,\lambda_2,\lambda_3) =
\Pa{u(\lambda_1-\lambda_3),v(\lambda_1+\lambda_3)},
(\lambda_1,\lambda_3,\lambda_2) =
\Pa{u(2\lambda_1+\lambda_3),-v\lambda_3},\\
(\lambda_2,\lambda_1,\lambda_3) =
\Pa{-u(\lambda_1+2\lambda_3),-v\lambda_1},(\lambda_2,\lambda_3,\lambda_1)
= \Pa{-u(2\lambda_1+\lambda_3),-v\lambda_3},
\\
(\lambda_3,\lambda_1,\lambda_2) =
\Pa{u(\lambda_1+2\lambda_3),-v\lambda_1},(\lambda_3,\lambda_2,\lambda_1)
= \Pa{-u(\lambda_1-\lambda_3),v(\lambda_1+\lambda_3)}.
\end{eqnarray*}
Then
\begin{eqnarray*}
\sum_{w\in S_3}(-1)^{l(w)}\delta_{w\lambda} &=&
\delta_{\Pa{u(\lambda_1-\lambda_3),v(\lambda_1+\lambda_3)}}+\delta_{\Pa{u(\lambda_1+2\lambda_3),-v\lambda_1}}+\delta_{\Pa{-u(2\lambda_1+\lambda_3),-v\lambda_3}}\\
&&- \delta_{\Pa{u(2\lambda_1+\lambda_3),-v\lambda_3}} -
\delta_{\Pa{-u(\lambda_1+2\lambda_3),-v\lambda_1}} -
\delta_{\Pa{-u(\lambda_1-\lambda_3),v(\lambda_1+\lambda_3)}}.
\end{eqnarray*}


\begin{prop}\label{lem:iter-convol}
The measure $H_{-\alpha_{12}}*H_{-\alpha_{13}}*H_{-\alpha_{23}}=
H_{-\bsh_1}*H_{-\bsh_1-\bsh_2}*H_{-\bsh_2}$ has Lebesgue density:
\begin{eqnarray}\label{eq:pdfmomentpolytope}
f(\nu) =\frac1{\sqrt{3}}\times
\begin{cases}0,&\text{if } \nu\in C_0;\\
-\nu_1,&\text{if
}\nu\in C_1;\\
\nu_3,&\text{if } \nu\in C_2.
\end{cases}
\end{eqnarray}
Here $C_0$ is the complementary region of the positive cone formed
by three vectors $\set{-\bsh_1,-\bsh_1-\bsh_2,-\bsh_2}$, $C_1$ is
the positive cone generated by $\set{-\bsh_1,-\bsh_1-\bsh_2}$, and
$C_2$ is the positive cone generated by
$\set{-\bsh_1-\bsh_2,-\bsh_2}$.
\end{prop}

\begin{proof}[The first proof]
In the present proof, we follow up the method used in
\cite{Christandl2014}, where Paradan's wall-crossing formula
\cite{Boysal2009} are heavily used. The measure
$H_{-\bsh_1}*H_{-\bsh_1-\bsh_2}*H_{-\bsh_2}$ is, in fact, the
non-Abelian Duistermaat-Heckman measure that is on the closures of
the regular chambers containing the vertex $(0,0)$ given by the
convolution
$\delta_{(0,0,0)}*H_{-\bsh_1}*H_{-\bsh_1-\bsh_2}*H_{-\bsh_2}\cong
\delta_{(0,0)}* H_{\Pa{-u,-v}}* H_{(-2u,0)}* H_{\Pa{-u,v}}$. Its
density is denoted by $f(\nu)\cong
f(u(\nu_1-\nu_3),v(\nu_1+\nu_3))$. Thus $\omega_1 =\Pa{-u,v},
\omega_2=\Pa{-2u,0},\omega_3=\Pa{-u,-v}$. \\
(i) Clearly $f\equiv0$ on $C_0$. The wall $W_{01}$ separating $C_0$
and $C_1$ is given by the equation:
$\frac{x}{-1}=\frac{y}{\sqrt{3}}$. Its normal vector
$\xi_{01}=(-\sqrt{3},-1)$. Just only one weights
$\omega_1=\Pa{-u,v}$ lies on the linear hyperplane spanned by
$W_{01}$ (other weights are outside of $W_{01}$:
$\omega_2=\Pa{-2u,0},\omega_3=\Pa{-u,-v}$). Consider the
push-forward of Lebesgue measure on $\real_{\geqslant0}$ along the
linear map $P_{W_{01}}:t\mapsto t\omega_1$. Its density with respect
to $\dif w$ is given by a single homogeneous polynomial on the wall
$W_{01}$. Denote by $f_{W_{01}}$ any polynomial function extending
it to all of $\mathrm{i}\liet^*$, the Lie algebra $\rS\rU(3)$.
Clearly $f_{W_{01}}=\sqrt{2}$. Indeed,
\begin{eqnarray*}
f^{-1}_{W_{01}}=\dif\lambda (\omega_1,\xi_{01}/\norm{\xi_{01}}^2) =
\norm{\begin{array}{cc}
        -u & v \\
        -\frac{\sqrt{3}}4 & -\frac14
      \end{array}
}=\frac1{\sqrt{2}}.
\end{eqnarray*}
Hence
\begin{eqnarray*}
f(\nu) =
\sqrt{2}\mathrm{Res}_{z=0}\Pa{\frac{e^{\Inner{\nu}{\bsx+z\xi_{01}}}}
{\Inner{\omega_2}{\bsx+z\xi_{01}}\Inner{\omega_3}{\bsx+z\xi_{01}}}}_{\bsx=0},
\end{eqnarray*}
Thus
\begin{eqnarray*}
f(\nu)
=\sqrt{2}\times\mathrm{Res}_{z=0}\Pa{\frac{\exp\Pa{-\sqrt{6}\nu_1z}}{6z^2}}.
\end{eqnarray*}
Therefore, on the chamber $C_1$, $ f(\nu) =
-\frac1{\sqrt{3}}\nu_1$.\\
(ii) The wall $W_{12}$ separating $C_1$
and $C_2$ is given by the equation: $y=0$. Its normal vector
$\xi_{12}=(0,-1)$. Just only one weights $\omega_2=(-2u,0)$ lies on
the linear hyperplane spanned by $W_{12}$ (other weights are outside
of $W_{12}$: $\omega_1=\Pa{-u,v},\omega_3=\Pa{-u,-v}$). Consider the
push-forward of Lebesgue measure on $\real_{\geqslant0}$ along the
linear map $P_{W_{12}}:t\mapsto t\omega_2$. Its density with respect
to $\dif w$ is given by a single homogeneous polynomial on the wall
$W_{12}$. Denote by $f_{W_{12}}$ any polynomial function extending
it to all of $\mathrm{i}\liet^*$, the Lie algebra $\rS\rU(3)$.
Clearly $f_{W_{12}}=\frac1{\sqrt{2}}$. Indeed,
\begin{eqnarray*}
f^{-1}_{W_{12}}=\dif\lambda (\omega_2,\xi_{12}/\norm{\xi_{12}}^2) =
\norm{\begin{array}{cc}
        -2u & 0 \\
        0 & -1
      \end{array}
}=\sqrt{2}.
\end{eqnarray*}
Hence
\begin{eqnarray*}
f(\nu) -  \Pa{-\frac1{\sqrt{3}}\nu_1}=
\frac1{\sqrt{2}}\mathrm{Res}_{z=0}\Pa{\frac{e^{\Inner{\nu}{\bsx+z\xi_{12}}}}
{\Inner{\omega_1}{\bsx+z\xi_{12}}\Inner{\omega_3}{\bsx+z\xi_{12}}}}_{\bsx=0}.
\end{eqnarray*}
Thus
\begin{eqnarray*}
f(\nu) + \frac1{\sqrt{3}}\nu_1=
\frac1{\sqrt{2}}\times\mathrm{Res}_{z=0}\Pa{\frac{\exp\Pa{-\sqrt{\frac32}(\nu_1+\nu_3)z}}{-\frac32z^2}}
\end{eqnarray*}
Therefore, on the chamber $C_2$, $f(\nu) = \frac1{\sqrt{3}}\nu_3$.
In summary, we obtain the desired identity
\eqref{eq:pdfmomentpolytope}. This completes the proof.
\end{proof}
If we denote by $(x,y)=\Pa{u(\nu_1-\nu_3),v(\nu_1+\nu_3)}$, then the
above result can be reformulated as
\begin{eqnarray}\label{eq:binary-f}
f(x,y) =
\begin{cases}\frac{\abs{y}-\sqrt{3}x}{3\sqrt{2}},&\text{if
}(x,y)\in\Set{(x,y)\in\real^2: x\leqslant0,\sqrt{3}x\leqslant
y\leqslant-\sqrt{3}x},\\
0,&\text{otherwise}.
\end{cases}
\end{eqnarray}
The support of this binary function in \eqref{eq:binary-f} and its
graph are depicted in Figure~\ref{fig:iter-supp}.
\begin{figure}[ht]\centering
\subfigure[Chambers and support (\emph{Shadow region})]
{\begin{minipage}[b]{0.47\linewidth}
\includegraphics[width=.7\textwidth]{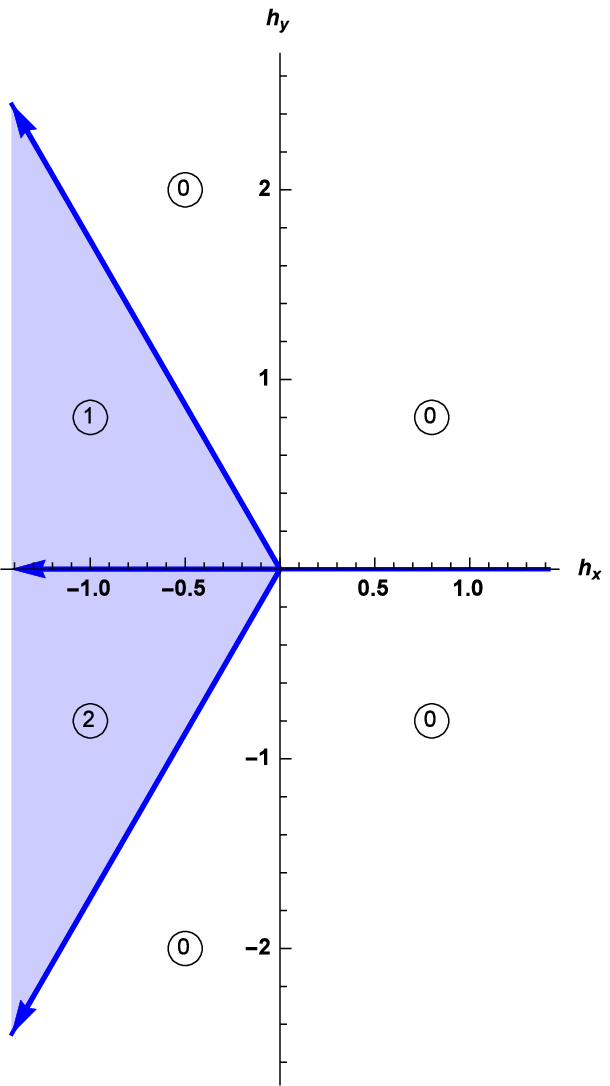}
\end{minipage}}\centering
\subfigure[The density function \eqref{eq:binary-f} of the iterated
convolution computed in Proposition~\ref{lem:iter-convol}]
{\begin{minipage}[b]{0.47\linewidth}
\includegraphics[width=1\textwidth]{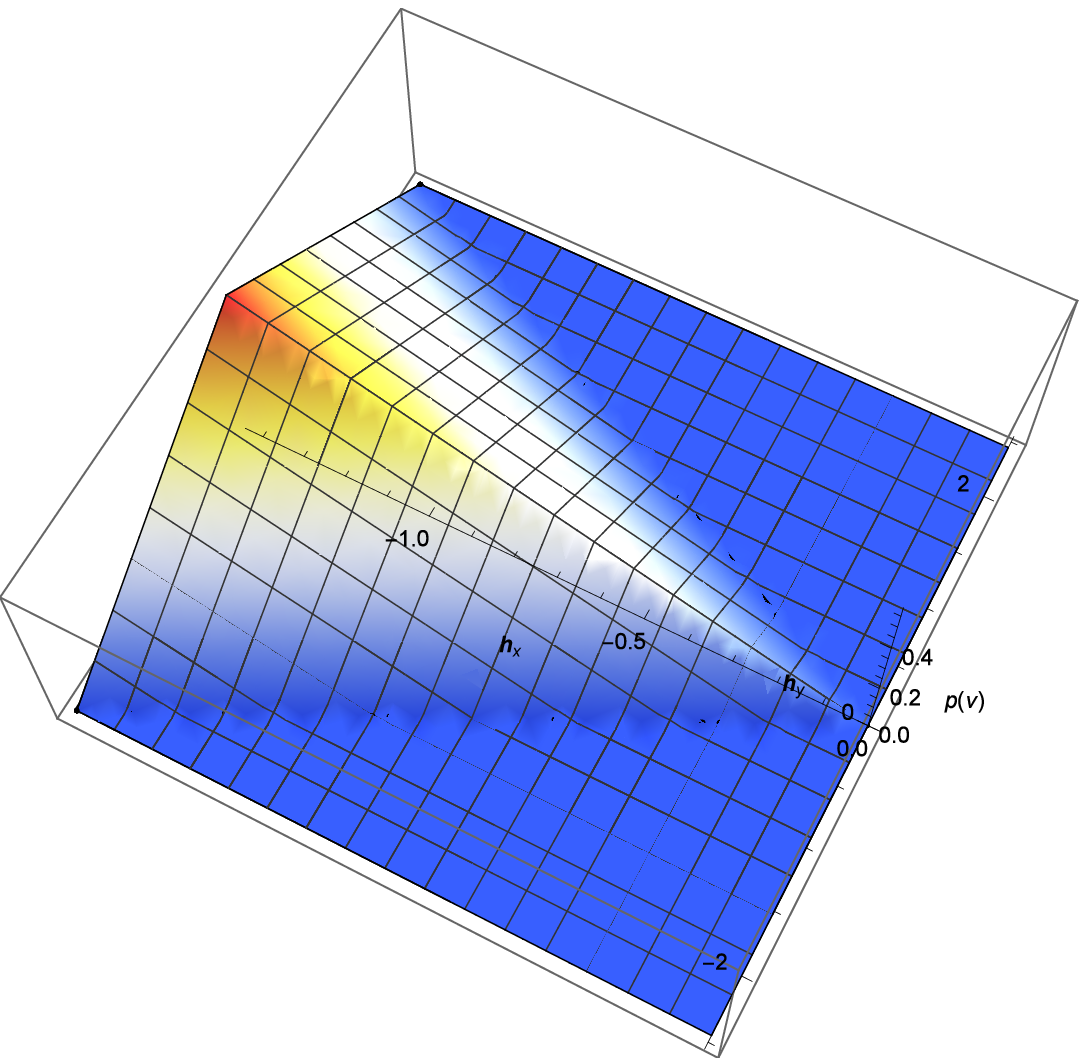}
\end{minipage}}
\caption{The support of the iterated convolution in
Proposition~\ref{lem:iter-convol} and its density over the support}
\label{fig:iter-supp}
\end{figure}

\begin{proof}[The second proof]
The measure $H_{-\bsh_1}*H_{-\bsh_1-\bsh_2}*H_{-\bsh_2}$ is
equivalently given by the convolution
$$
H_{\Pa{-u,-v}}* H_{(-2u,0)}* H_{\Pa{-u,v}}.
$$
Denote by $\sC$ the positive cone generated by three vectors
$(-u,-v),(-2u,0),(-u,v)$, i.e., $\sC_3:=C_1\cup C_2$ in
Figure~\ref{fig:iter-supp}, and $C_0=\real^2\backslash \sC_3$. We
can directly compute its density:
\begin{eqnarray*}
&&\int^\infty_0\dif t_1\int^\infty_0\dif t_2\int^\infty_0\dif t_3
\delta\Pa{t_1\Pa{\begin{array}{c}
                   -u \\
                   -v
                 \end{array}
}+t_2\Pa{\begin{array}{c}
           -u \\
           v
         \end{array}
}+t_3\Pa{\begin{array}{c}
           -2u \\
           0
         \end{array}
}-\Pa{\begin{array}{c}
        x \\
        y
      \end{array}
}}\\
&&= \int^\infty_0\dif t_1\int^\infty_0\dif t_2\int^\infty_0\dif t_3
\delta\Pa{\Pa{\begin{array}{c}
                   u(-t_1-t_2) \\
                   u\sqrt{3}(-t_1+t_2)
                 \end{array}
}+\Pa{\begin{array}{c}
           -2ut_3-x \\
           -y
         \end{array}
}}.
\end{eqnarray*}
By change of variables, i.e.,
$s_1=u(t_1+t_2),s_2=\sqrt{3}u(t_1-t_2)$, we get that the density is
given by
\begin{eqnarray*}
&&\frac1{\sqrt{3}}\int^\infty_0\dif s_1\int^\infty_{-\infty}\dif
s_2\int^\infty_0\dif t_3 \delta\Pa{\Pa{\begin{array}{c}
                   -s_1 \\
                   -s_2
                 \end{array}
}-\Pa{\begin{array}{c}
           2ut_3+x \\
           y
         \end{array}
}}\\
&&=\frac1{\sqrt{3}}\int^\infty_0\dif
t_3\mathbf{\I}_{\sC}\Pa{-\Pa{\begin{array}{c}
           2ut_3+x \\
           y
         \end{array}
}}
\end{eqnarray*}
Now $-\Pa{\begin{array}{c}
           2ut_3+x \\
           y
         \end{array}
}\in\sC$ if and only if
$\abs{y}\leqslant\sqrt{3}\abs{2ut_3+x}\leqslant\sqrt{3}x$ and
$x\leqslant2ut_3+x\leqslant -\frac{\abs{y}}{\sqrt{3}}$, i.e.,
$0\leqslant t_3\leqslant
\frac{\abs{y}}{\sqrt{6}}-\frac{x}{\sqrt{2}}$. Therefore the density
is given by
$\frac1{\sqrt{3}}\int^{\frac{\abs{y}}{\sqrt{6}}-\frac{x}{\sqrt{2}}}_0\dif
t_3=\frac{\abs{y}-\sqrt{3}x}{3\sqrt{2}}$. The support for this
density function is $\sC_3=\Set{(x,y)\in\real^2: x\leqslant0,
\sqrt{3}x\leqslant y\leqslant -\sqrt{3}x}$.
\end{proof}

\begin{exam}\label{exam:qutrit}
Let $u=\sqrt{\frac12}$ and $v=\sqrt{\frac32}$. Let
$\lambda_1=\frac16,\lambda_2=0,\lambda_3=-\frac16$ and
$\mu_1=\frac13,\mu_2=-\frac1{12},\mu_3=-\frac14$. The density of the
non-Abelian D-H measure is given by the restriction to the positive
Weyl chamber of an alternating sum of 36 copies of the density
described in Proposition~\ref{lem:iter-convol}. Note that the
geometry of the support of the density described in
Proposition~\ref{lem:iter-convol}, it is easily seen that only
summands for points in the shaded region, i.e.,
$\Set{(x,y)\in\real^2: x\geqslant0,-\sqrt{3}x\leqslant y\leqslant
\sqrt{3}x}$ contribute, see Figure~\ref{fig:qutrit-supp} (a).
Clearly the shaded region is a convex cone generated by three
positive roots $\Set{\alpha_{12},\alpha_{13},\alpha_{23}}$. From
Eq.~\eqref{eq:sum-qutrit}, we can pick out the following ten points
which are falling in such convex cone:{\small
\begin{eqnarray*}
\Pa{\frac{11u}{12},\frac
v{12}},\Pa{\frac{7u}{12},\frac{5v}{12}},\Pa{\frac
{7u}{12},\frac{v}{12}},\Pa{\frac{5u}{12},\frac
v4},\Pa{\frac{5u}{12},-\frac
v{12}},\\
\Pa{\frac{u}3,-\frac{v}6},\Pa{\frac{3u}4,\frac v4},\Pa{\frac
{3u}4,-\frac v{12}},\Pa{\frac u2,-\frac v3},\Pa{\frac u4,\frac
v{12}}.
\end{eqnarray*}}
Therefore,
\begin{eqnarray*}
\dh^K_{\cO_\lambda\times\cO_\mu}
&=&\left(\delta_{\Pa{\frac{11u}{12},\frac v{12}}}+
\delta_{\Pa{\frac{7u}{12},\frac{5v}{12}}} + \delta_{\Pa{\frac
{7u}{12},\frac{v}{12}}}+\delta_{\Pa{\frac{5u}{12},\frac
v4}}+\delta_{\Pa{\frac{5u}{12},-\frac v{12}}} +
\delta_{\Pa{\frac{u}3,-\frac{v}6}}\right.\\
&&\left.~~~- 2\delta_{\Pa{\frac{3u}4,\frac v4}} - \delta_{\Pa{\frac
{3u}4,-\frac v{12}}} -\delta_{\Pa{\frac u2,-\frac v3}}-
2\delta_{\Pa{\frac u4,\frac
v{12}}}\right)*H_{-\alpha_{12}}*H_{-\alpha_{13}}*H_{-\alpha_{23}}\big|_{\mathrm{i}\liet^*_{\geqslant0}}.
\end{eqnarray*}
Here the positive Weyl chamber $\mathrm{i}\liet^*_{\geqslant0}$ is
identified with $\Set{(x,y)\in \real^2:
x\geqslant0,-\frac1{\sqrt{3}}x\leqslant y\leqslant
\frac1{\sqrt{3}}x}$, see Figure~\ref{fig:qutrit-supp}(a).
\begin{figure}[ht]\centering
\subfigure[Chambers and support $\sS_3=\bigcup^6_{j=1}C_j$]
{\begin{minipage}[b]{0.47\linewidth}
\includegraphics[width=.7\textwidth]{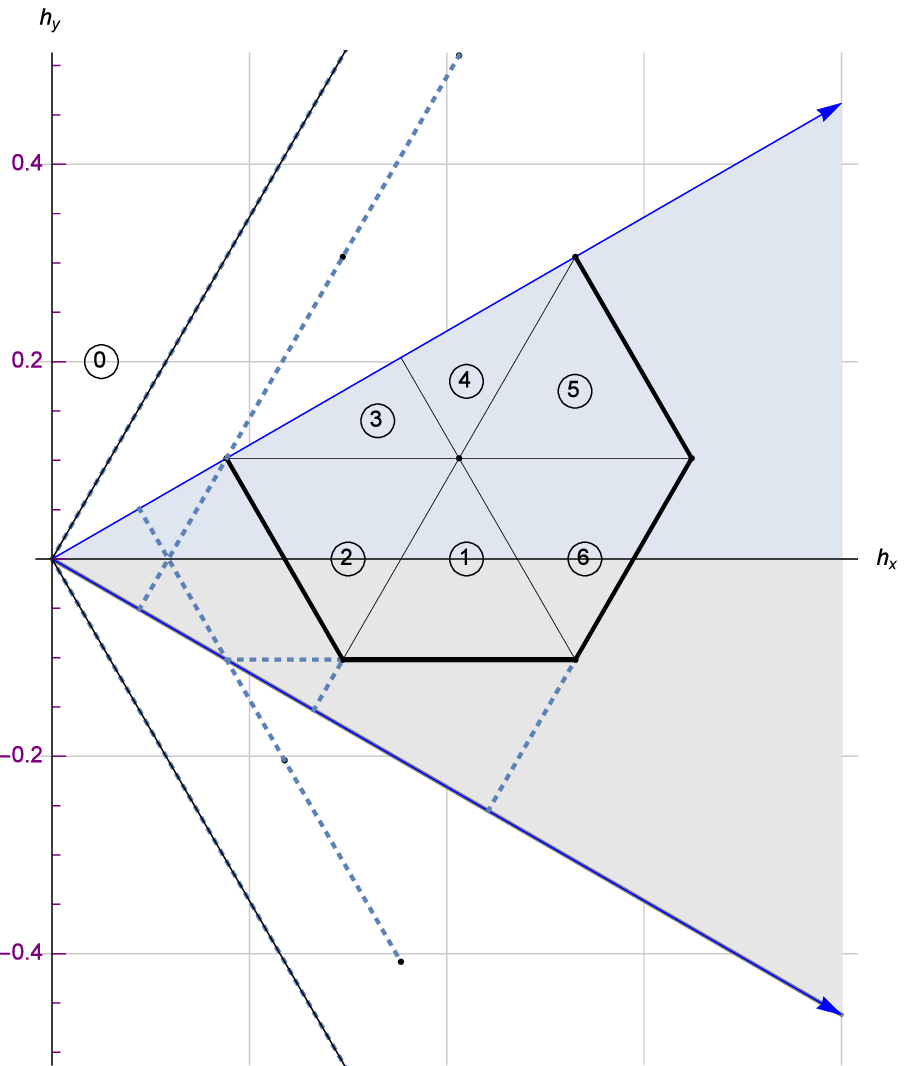}
\end{minipage}}\centering
\subfigure[The density function \eqref{eq:density-sum-twoqutrit} of
$(\tau_K)_*(\Phi_K)_*\Pa{\frac{\mu_M}{\vol(M)}}$]
{\begin{minipage}[b]{0.47\linewidth}
\includegraphics[width=1\textwidth]{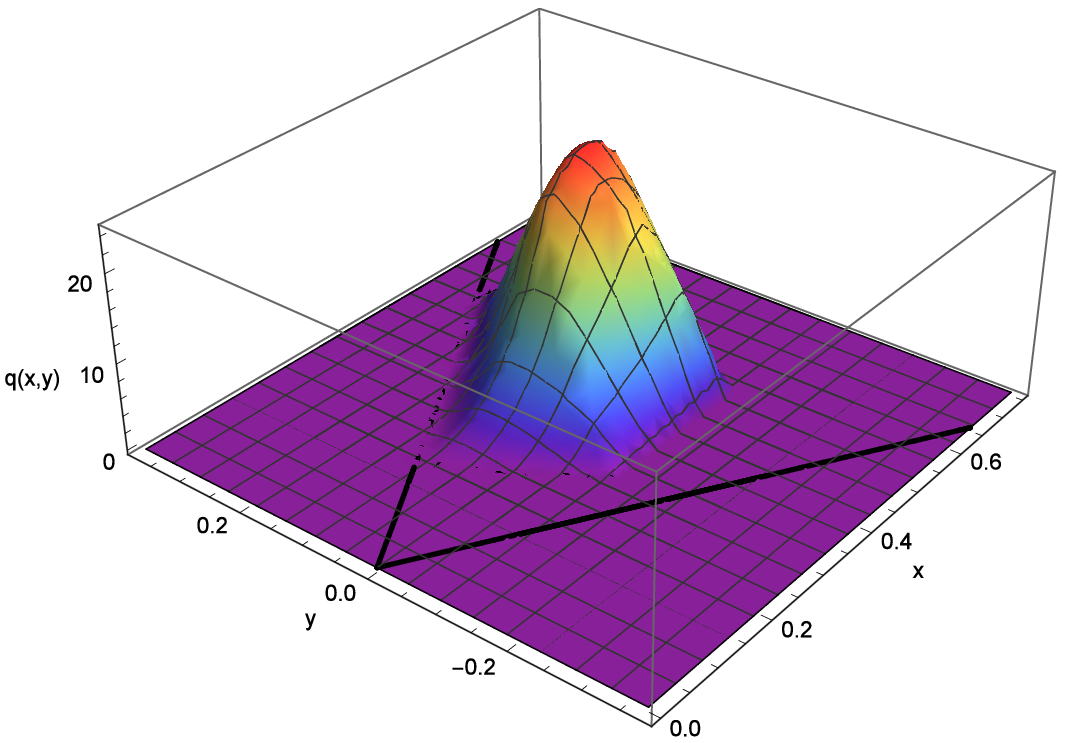}
\end{minipage}}
\caption{The support of the iterated convolution and its density
over the support for Example~\ref{exam:qutrit}}
\label{fig:qutrit-supp}
\end{figure}
The density of the non-Abelian D-H measure is given as
\begin{eqnarray*}
p(x,y)=\begin{cases}0,&\text{if }(x,y)\in C_0;\\
\frac{24 y+\sqrt{6}}{36 \sqrt{2}},&\text{if }(x,y)\in C_1;\\
\frac{6 \sqrt{3} x + 6 y-\sqrt{6}}{18\sqrt{2}},&\text{if }(x,y)\in C_2;\\
\frac{x-\sqrt{3}y}{\sqrt{6}},&\text{if }(x,y)\in C_3\cup C_4;\\
\frac{\sqrt{6} - 2 \sqrt{3} x - 2 y}{6 \sqrt{2}},&\text{if }(x,y)\in C_5;\\
\frac{5\sqrt{6} - 12 \sqrt{3}x + 12 y}{36 \sqrt{2}},&\text{if
}(x,y)\in C_6.
\end{cases}
\end{eqnarray*}
Here all $C_j$'s are the corresponding regions, marked by the
circled-numbers in Figure~\ref{fig:qutrit-supp}(a). We have already
known that
\begin{eqnarray*}
(\tau_K)_*(\Phi_K)_*\Pa{\frac{\mu_M}{\vol(M)}} =
p_K\frac{\dh^K_M}{\vol(M)},
\end{eqnarray*}
where $p_K(\nu) = \frac{(\nu_1-\nu_2)(\nu_1-\nu_3)(\nu_2-\nu_3)}2$.
Then $p_K(\nu)$, via $(x,y)=\Pa{u(\nu_1-\nu_3),v(\nu_1+\nu_3)}$, can
be rewritten as
\begin{eqnarray*}
p_K(x,y) = \frac{x^3 - 3 x y^2}{2 \sqrt{2}},\quad
(x,y)\in\Set{(x,y)\in\real^2:
x\geqslant0,-\frac1{\sqrt{3}}x\leqslant y\leqslant
\frac1{\sqrt{3}}x}.
\end{eqnarray*}
and $\vol(M) = \vol(\cO_\lambda)\vol(\cO_\mu) = p_K\Pa{\frac
u3,0}p_K\Pa{\frac{7u}{12},\frac v{12}} = \frac{35}{373248}$.
Therefore the density (see Figure~\ref{fig:qutrit-supp}(b)) for
$(\tau_K)_*(\Phi_K)_*\Pa{\frac{\mu_M}{\vol(M)}}$ (with the support
being $\sS_3:=\bigcup^6_{j=1}C_j$) is given as

\begin{eqnarray}\label{eq:density-sum-twoqutrit}
q(x,y) = \frac{2^53^4(x^3 - 3 x y^2)}{35}\times\begin{cases}0,&\text{if }(x,y)\in C_0;\\
24 y+\sqrt{6},&\text{if }(x,y)\in C_1;\\
12 \sqrt{3} x + 12 y-2\sqrt{6},&\text{if }(x,y)\in C_2;\\
12\sqrt{3}x-36y,&\text{if }(x,y)\in C_3\cup C_4;\\
6\sqrt{6} - 12 \sqrt{3} x - 12 y,&\text{if }(x,y)\in C_5;\\
5\sqrt{6} - 12 \sqrt{3}x + 12 y,&\text{if }(x,y)\in C_6.
\end{cases}
\end{eqnarray}
The normalization of the above function
\eqref{eq:density-sum-twoqutrit}, i.e., the integration over the
support being equal to one, is easily checked by computer system.
Consider the following random quantum state $\rho =
\frac{U\boldsymbol{\lambda}_1U^\dagger+V\boldsymbol{\lambda}_2V^\dagger}2$,
where $U,V$ are sampled by Haar measure over the unitary group and
$\boldsymbol{\lambda}_1=\diag(\frac12,\frac13,\frac16)$ and
$\boldsymbol{\lambda}_2=\diag(\frac23,\frac14,\frac1{12})$. Then
$2\Pa{\rho-\frac{\I_3}3} =
U\diag\Pa{\frac16,0,-\frac16}U^\dagger+V\diag\Pa{\frac13,-\frac1{12},-\frac14}V^\dagger$.
Let $s=\lambda_{\max}(\rho),t=\lambda_{\min}(\rho)$ be respective
maximal and minimal eigenvalues of $\rho$. Denote the eigenvalue
vector of random qubtrit $\rho$ by $(s,1-s-t,t)$, ordered
decreasingly. Thus the eigenvalue vector of $2(\rho-\I/3)$ is
$\Pa{2(s-1/3),2(2/3-s-t),2(t-1/3)}$. Its 2D coordinate via
$(\bsh_x,\bsh_y)$ is identified with $(x,y)$. Then
\begin{eqnarray*}
x=2u(s-t), y=2v\Pa{s+t-\frac23} \quad\text{or}\quad
s=\frac14\Pa{\frac yv+\frac xu}+\frac13, t=\frac14\Pa{\frac yv-\frac
xu}+\frac13.
\end{eqnarray*}
Therefore we can draw the conclusion that the joint density of
$(s,t)$ is given by
\begin{eqnarray}\label{eq:pdfforqutrit}
f(s,t) = \frac{2^{14}\cdot 3^6}{5\cdot 7}(s - t)\Pa{2s^2 + 2t^2 +
5st - 3s-3t+1}\cdot\Delta(s,t),
\end{eqnarray}
where
\begin{eqnarray*}
\Delta(s,t):=
\begin{cases}
s+2t-1,&\text{if } (s,t)\in\Set{(s,t): \frac5{12}\leqslant s\leqslant\frac7{12}, \max\Pa{\frac{17-24s}{24},\frac5{24}}\leqslant t\leqslant\frac{1-s}2};\\
\frac58-s-t,&\text{if } (s,t)\in\Set{(s,t):\frac5{12}\leqslant s\leqslant \frac12, \frac{5-8s}8\leqslant t\leqslant \frac5{24}};\\
\frac5{12}-s,&\text{if } (s,t)\in\Set{(s,t):\frac5{12} \leqslant s \leqslant\frac12, \frac5{24} \leqslant t \leqslant \frac{17 - 24s}{24}};\\
s-\frac7{12},&\text{if } (s,t)\in\Set{(s,t):\frac12 \leqslant s \leqslant \frac7{12}, \frac{17 - 24s}{24} \leqslant t \leqslant \frac5{24}};\\
\frac18-t,&\text{if } (s,t)\in\Set{(s,t): \frac12 \leqslant s
\leqslant \frac7{12}, \frac18 \leqslant t \leqslant \frac{17 -
24s}{24}}.
\end{cases}
\end{eqnarray*}
The support of the density function \eqref{eq:pdfforqutrit} are
plotted in the following Figure~\ref{fig:two-qutrit-supp}:
\begin{figure}[ht]\centering
\subfigure[Support] {\begin{minipage}[b]{.46\linewidth}
\includegraphics[width=.7\textwidth]{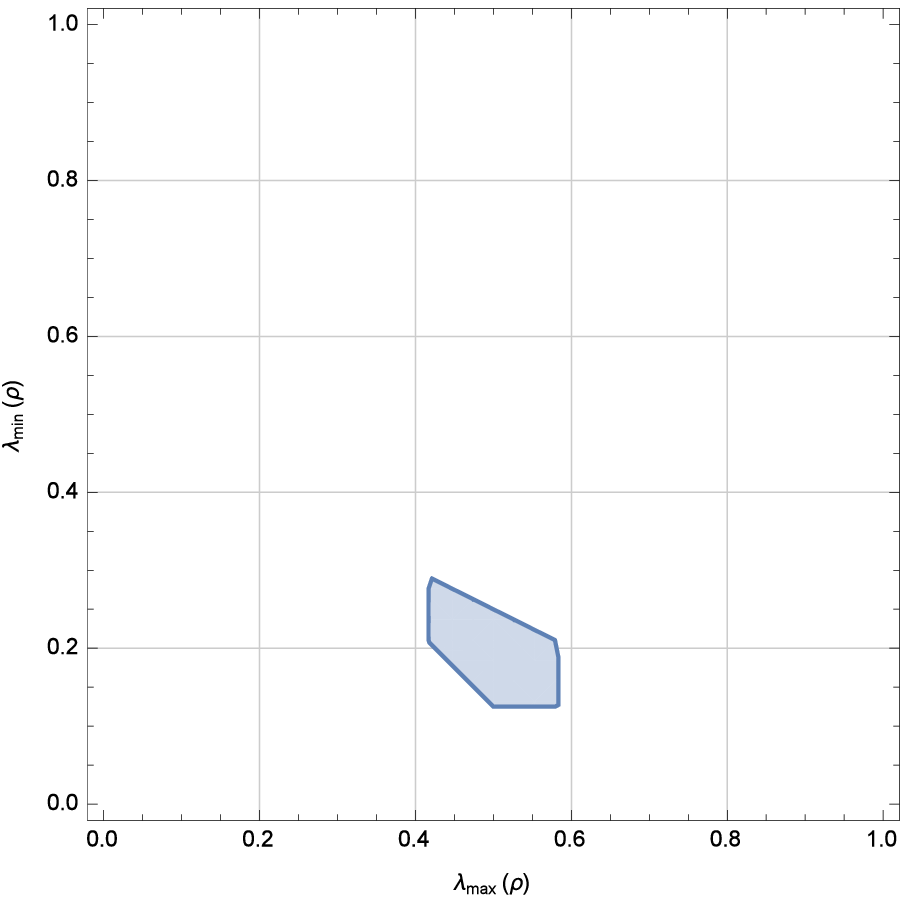}
\end{minipage}}\subfigure[Density function over the support] {\begin{minipage}[b]{.46\linewidth}
\includegraphics[width=1\textwidth]{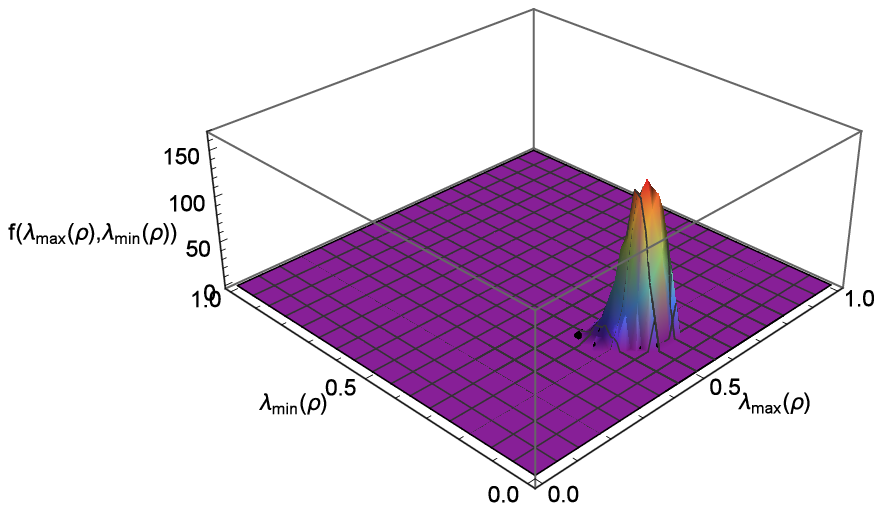}
\end{minipage}}\centering
\caption{(a) The feasible region of the maximal and minimal
eigenvalues of a random qutrit
$\rho=\frac12\Pa{U\boldsymbol{\lambda}_1U^\dagger+V\boldsymbol{\lambda}_2V^\dagger}$
where $\boldsymbol{\lambda}_1=\Pa{\frac12,\frac13,\frac16},
\boldsymbol{\lambda}_2=\Pa{\frac23,\frac14,\frac1{12}}$ and
$U,V\in\rS\rU(3)$ are Haar-distributed; (b) The joint density
function of the maximal and minimal eigenvalues of a random qutrit
$\rho=\frac12\Pa{U\boldsymbol{\lambda}_1U^\dagger+V\boldsymbol{\lambda}_2V^\dagger}$}
\label{fig:two-qutrit-supp}
\end{figure}
\end{exam}

\begin{remark}\label{rem:qutrit}
From the Example~\ref{exam:qutrit}, we see that the reasoning method
used in the Example~\ref{exam:qutrit} essentially provide a complete
solution to the joint density function of eigenvalues of the mixture
of quantum states algorithmically. That is, once two spectra are
given, we can give a analytical solution based on the given spectra.
Unfortunately, there is no unifying formula for that problem.
\end{remark}

\subsection{The mixture of two two-qubit states}

Given $K=\rS\rU(4)$. Then
$T=\Set{\diag(e^{\mathrm{i}\theta_1},e^{\mathrm{i}\theta_2},e^{\mathrm{i}\theta_3},e^{\mathrm{i}\theta_4}):\theta_1,\theta_2,\theta_3,\theta_4\in\real\bigwedge\sum^4_{j=1}\theta_j=0}$
with its Lie algebra
$\liet=\Set{\diag(\mathrm{i}\theta_1,\mathrm{i}\theta_2,\mathrm{i}\theta_3,\mathrm{i}\theta_4):\theta_1,\theta_2,\theta_3,\theta_4\in\real\bigwedge\sum^4_{j=1}\theta_j=0}$.
Thus $\bsh_1=\diag(1,-1,0,0)$, $\bsh_2=\diag(0,1,-1,0)$,
$\bsh_3=\diag(0,0,1,-1)$ is the basis of $\liet$, and $\liet=
\mathrm{i}\real\cdot \bsh_1\oplus \mathrm{i}\real\cdot \bsh_2\oplus
\mathrm{i}\real\cdot \bsh_3\cong\real^3$. All positive roots
$\alpha_{ij}\in\liet^*$ with $i<j$, where $i,j\in\set{1,2,3,4}$, are
given by
$\alpha_{12}\cong\bsh_1,\alpha_{13}\cong\bsh_1+\bsh_2,\alpha_{14}\cong\bsh_1+\bsh_2+\bsh_3,\alpha_{23}\cong\bsh_2,
\alpha_{24}\cong\bsh_2+\bsh_3, \alpha_{34}\cong\bsh_3$. Note that
the above realizations are via the Hilbert-Schmidt inner product.
Furthermore, we can construct an orthonormal basis for
$\mathrm{i}\liet$ as follows:
\begin{eqnarray}\label{eq:hxhyhz}
\bsh_x=\frac{\bsh_1}{\sqrt{2}},\quad
\bsh_y=\frac{\bsh_3}{\sqrt{2}},\quad
\bsh_z=\frac{\bsh_1+2\bsh_2+\bsh_3}{2}.
\end{eqnarray}
Denote also $u=\sqrt{\frac12}$. Then all positive roots
$\alpha_{ij}\in\liet^*$ with $i<j$, where $i,j\in\set{1,2,3,4}$, via
$(\bsh_x,\bsh_y,\bsh_z)$, are given by
$\alpha_{12}=(2u,0,0),\alpha_{13}=\Pa{u,-u,1},
\alpha_{14}=\Pa{u,u,1}, \alpha_{23} =\Pa{-u,-u,1},\alpha_{24} =
\Pa{-u,u,1}, \alpha_{34} = (0,2u,0)$. Besides, the Weyl group is
given by $W=S_4$. Let $\lambda\in\mathrm{i}\liet^*_{>0}$ and
$\mu\in\mathrm{i}\liet^*_{\geqslant0}$. So,
\begin{eqnarray*}
\dh^K_{\cO_\lambda\times\cO_\mu}= \Pa{\sum_{w,w'\in
S_4}(-1)^{l(w)+l(w')}\delta_{w\lambda+w'\mu}}*
H_{-\alpha_{12}}*H_{-\alpha_{13}}*H_{-\alpha_{14}}*
H_{-\alpha_{23}}*H_{-\alpha_{24}}*H_{-\alpha_{34}}\big|_{\mathrm{i}\liet^*_{\geqslant0}}.
\end{eqnarray*}
We see that
$\lambda=(\lambda_1,\lambda_2,\lambda_3,\lambda_4)\in\mathrm{i}\liet^*_{>0}$,
i.e., $\lambda_1>\lambda_2>\lambda_3>\lambda_4$ and
$\sum^4_{j=1}\lambda_j=0$. From this, we get that
$\lambda_1>0>\lambda_4, \lambda_1+\lambda_2>0>\lambda_3+\lambda_4$.
With the orthonormal basis $\set{\bsh_x,\bsh_y,\bsh_z}$ of
$\mathrm{i}\liet$,
$\lambda\cong\Pa{u(\lambda_1-\lambda_2),-u(\lambda_1+\lambda_2+2\lambda_4),\lambda_1+\lambda_2}$,
where $u(\lambda_1-\lambda_2)>0,\lambda_1+\lambda_2>0$. In the
following, we calculate the density of the measure
$H_{-\alpha_{12}}*H_{-\alpha_{13}}*H_{-\alpha_{14}}*
H_{-\alpha_{23}}*H_{-\alpha_{24}}*H_{-\alpha_{34}}$.
\begin{prop}\label{lem:iter-convol-su4}
The measure $H_{-\alpha_{12}}*H_{-\alpha_{13}}*H_{-\alpha_{14}}*
H_{-\alpha_{23}}*H_{-\alpha_{24}}*H_{-\alpha_{34}}$ has Lebesgue
density:
\begin{eqnarray}\label{eq:pdf3}
f(\nu) = -\frac18(\nu_1+\nu_2)(3\nu_1+\nu_2)(\nu_1+\nu_2-2\nu_4).
\end{eqnarray}
\end{prop}

\begin{proof}
Now the measure $H_{-\alpha_{12}}*H_{-\alpha_{13}}*H_{-\alpha_{14}}*
H_{-\alpha_{23}}*H_{-\alpha_{24}}*H_{-\alpha_{34}}$ is represented
by (via $(\bsh_x,\bsh_y,\bsh_z)$)
$$
\delta_{(0,0,0)}*H_{(-2u,0,0)}*H_{(-u,u,-1)}*H_{(-u,-u,-1)}*
H_{(u,u,-1)}*H_{(u,-u,-1)}*H_{(0,-2u,0)}.
$$
Note that
$t_1(-2u,0,0)+t_2(0,-2u,0)+t_3(u,u,-1)+t_4(-u,u,-1)+t_5(-u,-u,-1)+t_6(u,-u,-1)$
is equal to
\begin{eqnarray*}
(u (- 2t_1 + t_3), u(- 2t_2 + t_3), - t_3) +(u (- t_4 - t_5 + t_6),
u(t_4 - t_5 - t_6), - t_4 - t_5 - t_6).
\end{eqnarray*}
Its density of the above measure can be computed in the following:
\begin{eqnarray*}
\iiint\iiint_{\real^6_+}\dif t_1\dif t_2\dif t_3\dif t_4\dif t_5\dif
t_6\delta\Pa{\Pa{\begin{array}{c}
           u (- 2t_1 + t_3) \\
          u(- 2t_2 + t_3) \\
          - t_3
        \end{array}
}+\Pa{\begin{array}{c}
           u (- t_4 - t_5 + t_6) \\
          u(t_4 - t_5 - t_6) \\
          - t_4 - t_5 - t_6
        \end{array}
}-\Pa{\begin{array}{c}
        x \\
        y \\
        z
      \end{array}
}},
\end{eqnarray*}
where $\nu=x\bsh_x+y\bsh_y+z\bsh_z\cong(x,y,z)$. Then performing
change of variables $s_1=u (2t_1 - t_3), s_2=u(2t_2 - t_3),
s_3=t_3$, we get the density
\begin{eqnarray*}
&&\frac12\iiint\iiint\dif s_1\dif s_2\dif s_3\dif t_4\dif t_5\dif
t_6\delta\Pa{\Pa{\begin{array}{c}
           -s_1 \\
          -s_2 \\
          -s_3
        \end{array}
}-\Pa{\begin{array}{c}
           u (t_4 + t_5 - t_6)+x \\
          u(-t_4 + t_5 + t_6)+y \\
          t_4 + t_5 + t_6+z
        \end{array}
}}\\
&&=\frac12\iiint\dif t_4\dif t_5\dif t_6
\mathbf{\I}_{\sD}\Pa{\begin{array}{c}
           u (t_4 + t_5 - t_6)+x \\
          u(-t_4 + t_5 + t_6)+y \\
          t_4 + t_5 + t_6+z
        \end{array}
},
\end{eqnarray*}
where $\sD$ is the positive cone generated by
$(-2u,0,0),(0,-2u,0),(u,u,-1)$. Now
$$
\Pa{u (t_4 + t_5 - t_6)+x, u(-t_4 + t_5 + t_6)+y, t_4 + t_5 +
t_6+z}\in\sD
$$
if and only if
\begin{eqnarray}\label{eq:region456}
\begin{cases}
0\leqslant t_j,j=4,5,6;\\
0\leqslant t_4+t_5+t_6\leqslant \abs{z};\\
0\leqslant 3t_4+3t_5+t_6\leqslant \sqrt{2}(\sqrt{2}\abs{z}-x);\\
0\leqslant t_4+3t_5+3t_6\leqslant \sqrt{2}(\sqrt{2}\abs{z}-y).
\end{cases}
\end{eqnarray}
where
\begin{eqnarray}\label{eq:constraints-xyz}
z\leqslant0,x\leqslant \sqrt{2}\abs{z}, y\leqslant \sqrt{2}\abs{z}.
\end{eqnarray}
Let $R(x,y,z)$ be the region determined by \eqref{eq:region456}
whenever $x,y,z$ satisfying \eqref{eq:constraints-xyz}. Then the
density is given by
$$
\frac12\iiint_{R(x,y,z)}\dif t_4\dif t_5\dif
t_6=\frac12\vol(R(x,y,z)).
$$
Performing change of variables $\alpha=t_4+t_5+t_6,
\beta=3t_4+3t_5+t_6, \gamma=t_4+3t_5+3t_6$, we obtain that $\dif
t_4\dif t_5\dif t_6=\frac14\dif \alpha\dif \beta\dif \gamma$.
Furthermore, we see that
\begin{eqnarray*}
\frac12\iiint_{R(x,y,z)}\dif t_4\dif t_5\dif t_6
&=&\frac12\vol(R(x,y,z))= \frac18\int^{\abs{z}}_0\dif
\alpha\int^{\sqrt{2}(\sqrt{2}\abs{z}-x)}_0\dif
\beta\int^{\sqrt{2}(\sqrt{2}\abs{z}-y)}_0\dif \gamma\\
&=&\frac{\abs{z}(\sqrt{2}\abs{z}-x)(\sqrt{2}\abs{z}-y)}4 =-\frac14
z(x+\sqrt{2}z)(y+\sqrt{2}z).
\end{eqnarray*}
Therefore the measure $H_{(-2u,0,0)}*H_{(-u,u,-1)}*H_{(-u,-u,-1)}*
H_{(u,u,-1)}*H_{(u,-u,-1)}*H_{(0,-2u,0)}$ has the density:
$$
f(x,y,z)=-\frac14 z(x+\sqrt{2}z)(y+\sqrt{2}z),
$$
where $(x,y,z)\in\sC_4:=\Set{(x,y,z)\in\real^3:
z\leqslant0,x+\sqrt{2}z\leqslant0,y+\sqrt{2}z\leqslant0}$. Via
$x=u(\nu_1-\nu_2),y=-u(\nu_1+\nu_2+2\nu_4),z=\nu_1+\nu_2$, we get
Eq.~\eqref{eq:pdf3}. This completes the proof.
\end{proof}

\begin{remark}
The density of the non-Abelian D-H measure is given by the
restriction to the positive Weyl chamber of an alternating sum of
$(4!)^2=576$ copies of the density described in
Proposition~\ref{lem:iter-convol-su4}. Note that the geometry of the
support of the density described in
Proposition~\ref{lem:iter-convol-su4}, it is easily seen that only
summands for points in the region, i.e., $\Set{(x,y,z)\in\real^3:
z\geqslant0,x+\sqrt{2}z\geqslant0,y+\sqrt{2}z\geqslant0}$,
contribute. Clearly this region is a convex cone generated by six
positive roots $\Set{\alpha_{ij}:i<j;i,j\in[4]}$. Based on this, we
can get that the positive Weyl chamber
$\mathrm{i}\liet^*_{\geqslant0}$ is identified with
$\Set{(x,y,z)\in\real^3: x\geqslant0,y\geqslant0,z\geqslant
u(x+y)}$. Theoretically, the specific computation of such density is
workable, see Remark~\ref{rem:qutrit}. But the large number of terms
leads to the increase of computational complexity of this problem,
even running in computer. We remark here that the spectral density
of the mixture of $4\times4$ density matrices can be applied to
study entanglement of a random two-qubit since $4\times4$ density
matrices can be viewed as states of two-qubits. We are encouraging
those people who are interested in this problem. In the following,
as an illustration, we compute the density function over some
subregion of the support of the density function.
\end{remark}

\begin{exam}\label{exam:qu4it}
Denote still $u=\sqrt{\frac12}$. Now each 4-tuple
$\nu=(\nu_1,\nu_2,\nu_3,\nu_4)$, via $(\bsh_x,\bsh_y,\bsh_z)$, is
represented as $\Pa{u(\nu_1-\nu_2), -u(\nu_1+\nu_2+2\nu_4),
\nu_1+\nu_2}:=(x,y,z)$. Let
$\lambda=(\lambda_1,\lambda_2,\lambda_3,\lambda_4)$, where
$\lambda_1=\frac18,\lambda_2=\frac1{12},\lambda_3=-\frac1{24},\lambda_4=-\frac16$
and $\mu=(\mu_1,\mu_2,\mu_3,\mu_4)$, where
$\mu_1=\frac38,\mu_2=\frac1{64},\mu_3=-\frac{11}{64},\mu_4=-\frac7{32}$.
Let $P_0$ be the positive cone generated by three vectors
$\set{\bsx_1=(2u,0,0),\bsx_2=(0,2u,0),\bsx_3=(-u,-u,1)}$, i.e.,
\begin{eqnarray*}
P_0&:=&\real_{\geqslant0}\bsx_1+\real_{\geqslant0}\bsx_2+\real_{\geqslant0}\bsx_3\\
&=&\Set{r_1\bsx_1+r_2\bsx_2+r_3\bsx_3:r_i\in\real_{\geqslant0},i=1,2,3}.
\end{eqnarray*}
And we also let $Q_0:=-P_0$. Denote by $\bsv_j(j=1,\ldots,198)$ the
above points or vectors. Let $Q_j=Q_0+\bsv_j, j=1,\ldots,198$, where
the sum is taken as Minkowski sum, defined as
$A+\bsu:=\set{\bsa+\bsu:\bsa\in A}$. In addition, the region in the
positive Weyl chamber determined by the eigenvalue-vector of a
two-qubit state is identified as
$$
V_0:=\Set{(x,y,z)\in\real^3_+:
0<x<u,0<y<\frac{u-x}3,u(x+y)<z<\frac{1-4uy}2}.
$$
Now we can compute the density function over the following
subregion:{\small
\begin{eqnarray}\label{eq:3Dregion}
V_0\bigcap \Pa{P_0 \bigcap^{198}_{j=1} Q_j} =
\Set{(x,y,z):0<x<\frac{3u}{128},0<y<\frac{6u-256
x}{128},u(x+y)<z\leqslant \frac{3-128u x}{128}}.
\end{eqnarray}}
As an illustration, with the help of \textsc{Mathematica}~10.0, we
compute the probability distribution over the above region
\eqref{eq:3Dregion} that {\scriptsize
\begin{eqnarray*}
f(x,y,z) = \frac{9535}{9437184} - \frac{91x-73y}{24576 \sqrt{2}} -
\frac{13xy}{128} - \frac{487z}{147456} -
\frac{(71x+91y)z}{384\sqrt{2}} + \frac{xyz}2 - \frac{7z^2}{32} +
\frac{(x+y)z^2}{\sqrt{2}} +z^3,
\end{eqnarray*}}
where $(x,y,z)\in V_0\bigcap \Pa{P_0 \bigcap^{198}_{j=1} Q_j}$ and
thus the corresponding density function of $\nu$, as the vector of
eigenvalues of
$U\boldsymbol{\lambda}U^\dagger+V\boldsymbol{\mu}V^\dagger=2\Pa{\rho-\frac{\I_4}4}$,
is given as {\scriptsize
\begin{eqnarray*}
&&f(\nu)\equiv f(\nu_1,\ldots,\nu_4)\\
&&= \frac1{9437184}\left(9535 - 62656\nu_1  - 1339392\nu^2_1 +
7077888\nu^3_1 - 27712\nu_2 -
1892352\nu_1\nu_2+16515072\nu^2_1\nu_2-552960\nu^2_2+11796480\nu_1\nu^2_2\right.\\
&&\left.~~~~~~~~~~~~~~~~~~+2359296\nu^3_2-28032\nu_4+3194880\nu_1\nu_4-14155776\nu^2_1\nu_4+1277952\nu_2\nu_4-18874368\nu_1\nu_2\nu_4-4718592\nu^2_2\nu_4\right),
\end{eqnarray*}}
where $\nu_1+\nu_2+\nu_3+\nu_4=0$. Since we have already known that
the support of the density is contained in $V_0$, it follows that
the density over the other subregion of the support can also be
computed analogously. We are not going to continue this topic here.
\end{exam}

\section{An application in quantum information
theory}\label{sect:app}

As already known, in order to quantify quantum coherence existing in
quantum states, Baumgratz \emph{et al} proposed that any
non-negative function $\sC$, defined over the state space, should
satisfy three properties \cite{Baumgratz2014}, one of which is the
property that \emph{the coherence measure should be non-increasing
under mixing of quantum states}, that is, it should be convex.
Clearly the rationality of such requirement is from physical
motivation. We shall make an attempt to explain why 'mixing reduces
coherence' statistically. Because the coherence measure can be
defined for many different ways. Here we take the so-called
\emph{relative entropy of coherence} as the coherence measure. As
noted in \cite{Lin2018}, the relative entropy of coherence is a
well-defined measure of coherence and satisfies all the required
properties of coherence. Moreover the relative entropy of coherence
has also a novel operational interpretation in terms of hypothesis
testing \cite{Berta2017}.

Now we can use our results, obtained in this paper, to give some
hints to the intuition in which the quantum coherence
\cite{Baumgratz2014} decreases statistically as the mixing times $n$
increasing in the equiprobable mixture of $n$ qudits. The
mathematical definition of the relative entropy of coherence can be
given as $\sC_r(\rho):=\rS(\rho^D)-\rS(\rho)$, where $\rho^D$ is the
diagonal part of the quantum state $\rho$ with respect to a prior
fixed orthonormal basis, and $\rS(\rho):=-\Tr{\rho\ln\rho}$ is the
von Neumann entropy. Denote $\overline{\sC}^{(n)}_r$ the average
coherence of the equiprobable mixture of $n$ i.i.d. random quantum
states from the Hilbert-Schmidt ensemble, i.e.,
$\overline{\sC}^{(n)}_r = \overline{\rS}^D_n-\overline{\rS}_n$,
where
\begin{eqnarray*}
\overline{\rS}_n:=\mathbb{E}_{\rho_1,\ldots,\rho_n\in\cE_{d,k}}\Br{\rS\Pa{\frac{\rho_1+\cdots+\rho_n}n}},\quad\overline{\rS}^D_n:=\mathbb{E}_{\rho_1,\ldots,\rho_n\in\cE_{d,k}}\Br{\rS\Pa{\frac{\rho^D_1+\cdots+\rho^D_n}n}}.
\end{eqnarray*}

\begin{prop}\label{prop:2-level}
(i) Assume that $r_1,r_2\in(0,1)$. The average entropy,
$\overline{\rS}_2(r_1,r_2)$, of the equiprobable mixture of two
random density matrices chosen from orbits $\cO_{\frac{1-r_1}2}$ and
$\cO_{\frac{1-r_2}2}$, respectively, is given by the formula:
\begin{eqnarray*}
\overline{\rS}_2(r_1,r_2)=\int^{r_+}_{r_-}
\rH_2\Pa{\frac{1-r}{2}}p(r|r_1,r_2)\dif r,
\end{eqnarray*}
where $\rH_2(x):=-x\log_2x-(1-x)\log_2(1-x)$ is the so-called binary
entropy function, and $p(r|r_1,r_2)$ is taken from \eqref{eq:rr1r2}.
Furthermore, we have
\begin{eqnarray*}
\overline{\rS}_2&:=&\mathbb{E}_{\rho_1,\rho_2\in\cE_{2,2}}\Br{\rS\Pa{\frac{\rho_1+\rho_2}2}}=\int^1_0\int^1_0\overline{\rS}_2(r_1,r_2)3r^2_13r^2_2\dif
r_1\dif r_2 \\
&=&\int^1_0 \rH_2\Pa{\frac{1-r}{2}} p^{(2)}(r)\dif r= \frac{221}{140
\ln2}-\frac{53}{35}\simeq 0.763111.
\end{eqnarray*}
(ii) The average entropy, $\overline{\rS}^D_2(a,b)$, of the diagonal
part of the equiprobable mixture of two random density matrices
chosen from orbits $\cO_a$ and $\cO_b$, where $a,b\in(0,\frac12)$,
is given by the formula:
\begin{eqnarray*}
\overline{\rS}^D_2(a,b)=\int^{1-t_0}_{t_0}\rH_2(x)q(x|a,b)\dif x,
\end{eqnarray*}
where $t_0=\frac{a+b}2$ and $t_1=\frac{1-\abs{a-b}}2$ for given $a$
and $b$. Furthermore, we have
\begin{eqnarray*}
\overline{\rS}^D_2:=\int^1_0\int^1_0\overline{\rS}^D_2\Pa{\frac{1-r_1}2,\frac{1-r_2}2}3r^2_13r^2_2\dif
r_1\dif r_2\simeq 0.92414.
\end{eqnarray*}
Therefore, we have that $\overline{\sC}^{(2)}_r=\overline{\rS}^D_2 -
\overline{\rS}_2 \simeq 0.16103$.
\end{prop}

\begin{proof}
Let $\rho_1\in \cO_{\frac{1-r_1}2}$ and $\rho_2\in
\cO_{\frac{1-r_2}2}$. By using Bloch representation
\eqref{eq:Bloch-rep}, we can rewrite then as $\rho_1=\rho(\bsr_1)$
and $\rho_2=\rho(\bsr_2)$, respectively, where $r_1=\abs{\bsr_1}$
and $r_2=\abs{\bsr_2}$. Thus
$\rho(\bsr)=\frac12(\rho(\bsr_1)+\rho(\bsr_2))$. We see that the von
Neumann entropy of the qubit $\rho(\bsr)$ is given by
\begin{eqnarray}
\rS(\rho(\bsr)) = \rH_2\Pa{\frac{1-r}2} =
-\frac{1+r}2\log_2\frac{1+r}2-\frac{1-r}2\log_2\frac{1-r}2.
\end{eqnarray}
It is easily seen that the average entropy of the equiprobable
mixture of two random density matrices with given spectra is denoted
by $\overline{\rS}_2(r_1,r_2)$, which is given by
\begin{eqnarray}
\overline{\rS}_2(r_1,r_2) = \iint
\dif\mu_{\mathrm{Haar}}(U)\dif\mu_{\mathrm{Haar}}(V)\rS\Pa{\frac{U\rho(\bsr_1)U^\dagger+V\rho(\bsr_2)V^\dagger}2}.
\end{eqnarray}
Here $U,V$ are in $\rS\rU(2)$, and $\mu_{\mathrm{Haar}}$ is the
normalized Haar measure over the special unitary group $\rS\rU(2)$.
We see from Proposition~\ref{cor:rr1r2} that
\begin{eqnarray}
\overline{\rS}_2(r_1,r_2) = \int^{r_+}_{r_-}
\rH_2\Pa{\frac{1-r}{2}}p(r|r_1,r_2)\dif r,
\end{eqnarray}
where $r_-=\frac{r_1-r_2}2$ and $r_+=\frac{r_1+r_2}2$. Furthermore,
for $\rho_1,\rho_2$ chosen independently in $\cE_{2,2}$ by
\eqref{eq:Bloch-lenghth}, we have that the average entropy of the
mixture $\rho=\frac{\rho_1+\rho_2}2$ is given
\begin{eqnarray*}
\overline{\rS}_2:=\mathbb{E}_{\rho_1,\rho_2\in\cE_{2,2}}\Br{\rS\Pa{\frac{\rho_1+\rho_2}2}}
=\int^1_0\int^1_0\overline{\rS}_2(r_1,r_2)p(r_1)p(r_2)\dif r_1\dif
r_2\simeq0.76311.
\end{eqnarray*}
Here $p(r_1)=3r^2_1$ and $p(r_2)=3r^2_2$. Next the average entropy
of diagonal of mixture of two random density matrices for qubits is
directly obtained. Therefore, we have the desired conclusion:
$\overline{\sC}^{(2)}_r=\overline{\rS}^D_2 - \overline{\rS}_2 \simeq
0.16103$ \footnote{All the numerical values in the paper are
approximately computed by the computer software \textsc{Mathematica}
10.}.
\end{proof}

\begin{prop}\label{prop:3-level}
Assume that $r_i\in(0,1)(i=1,2,3)$. (i) The average entropy,
$\overline{\rS}_3(r_1,r_2,r_3)$, of the equiprobable mixture of
three random density matrices chosen from orbits
$\cO_{\frac{1-r_1}2}$, $\cO_{\frac{1-r_2}2}$, and
$\cO_{\frac{1-r_3}2}$, respectively, is given by the formula:
\begin{eqnarray*}
\overline{\rS}_3(r_1,r_2,r_3)=\int^{\frac{r_1+r_2+r_3}3}_0
\rH_2\Pa{\frac{1-r}{2}}p(r|r_1,r_2,r_3)\dif r,
\end{eqnarray*}
where $\rH_2(x)$ is the binary entropy function mentioned
previously, and $p(r|r_1,r_2,r_3)$ is taken from \eqref{eq:p(r)}.
Furthermore, we have
\begin{eqnarray*}
\overline{\rS}_3&:=&\mathbb{E}_{\rho_1,\rho_2,\rho_3\in\cE_{2,2}}\Br{S\Pa{\frac{\rho_1+\rho_2+\rho_3}3}}
=\int^1_0\int^1_0\int^1_0\overline{S}_3(r_1,r_2,r_3)3r^2_13r^2_23r^2_3\dif
r_1\dif r_2\dif r_3\\
&=&\int^1_0 \rH_2\Pa{\frac{1-r}{2}} p^{(3)}(r)\dif r = \frac{57821 +
94464 \ln2 - 105498 \ln3}{12600 \ln2}\simeq0.84696.
\end{eqnarray*}
(ii) The average entropy, $\overline{\rS}^D_3(r_1,r_2,r_3)$, of the
diagonal part of the equiprobable mixture of three random density
matrices chosen from orbits
$\cO_{\frac{1-r_1}2},\cO_{\frac{1-r_2}2}$ and $\cO_{\frac{1-r_3}2}$,
respectively, is given by the formula:
\begin{eqnarray*}
\overline{\rS}^D_3(r_1,r_2,r_3)=\int^1_0\rH_2(x)q(x|r_1,r_2,r_3)\dif
x
\end{eqnarray*}
Furthermore, we have
\begin{eqnarray*}
\overline{\rS}^D_3:=\int^1_0\int^1_0\int^1_0\overline{\rS}^D_3(r_1,r_2,r_3)3r^2_13r^2_23r^2_3\dif
r_1\dif r_2\dif r_3\approx0.95026.
\end{eqnarray*}
Therefore, we see that $\overline{\sC}^{(3)}_r=\overline{\rS}^D_3 -
\overline{\rS}_3 \simeq0.10329$.
\end{prop}

\begin{proof}
Conceptually, the idea of the proof is quite simple. Thus we omit it
here.
\end{proof}

Similarly, we get also that
\begin{eqnarray*}
\overline{\rS}_4= \int^1_0 \rH_2\Pa{\frac{1-r}{2}} p^{(4)}(r)\dif
r=\frac{22469023 + 25336464\ln2 - 35429400\ln3}{1801800
\ln2}\simeq0.886969.
\end{eqnarray*}

For any natural number $n\geqslant2$ and
$\Upsilon:=\sum^n_{j=1}\rho_j$, where $\rho_j$'s are independent and
identical distribution (i.i.d.) chosen from $\cE_{d,d}$, the
Hilbert-Schmidt ensemble, we have already known that
\cite{Zhang2018}:
\begin{eqnarray}\label{eq:eigen-part}
\mathbb{E}_{\rho_1,\ldots,\rho_n\in\cE_{d,d}}\Br{\rS\Pa{\frac{\sum^n_{j=1}\rho_j}n}}\geqslant\mathbb{E}_{\rho_1,\ldots,\rho_{n-1}\in\cE_{d,d}}\Br{\rS\Pa{\frac{\sum^{n-1}_{j=1}\rho_j}{n-1}}}.
\end{eqnarray}
By using the technique in the proof of the above inequality, we show
next that
\begin{eqnarray}\label{eq:diag-part}
\mathbb{E}_{\rho_1,\ldots,\rho_n\in\cE_{d,d}}\Br{\rS\Pa{\frac{\sum^n_{j=1}\rho^D_j}n}}\geqslant\mathbb{E}_{\rho_1,\ldots,\rho_{n-1}\in\cE_{d,d}}\Br{\rS\Pa{\frac{\sum^{n-1}_{j=1}\rho^D_j}{n-1}}}.
\end{eqnarray}
Indeed,
$\frac{\Upsilon}n=\frac1n\Pa{\sum^n_{j=1}\frac{\Upsilon-\rho_j}{n-1}}$.
Furthermore, its diagonal part is given by
$\frac{\Upsilon^D}n=\frac1n\Pa{\sum^n_{j=1}\frac{\Upsilon^D-\rho^D_j}{n-1}}$.
Due to the concavity of von Neumann entropy, we see that
\begin{eqnarray*}
\rS\Pa{\frac{\rho^D_1+\cdots+\rho^D_n}{n}}\geqslant
\frac1n\sum^n_{j=1}\rS\Pa{\frac{\Upsilon^D-\rho^D_j}{n-1}}.
\end{eqnarray*}
Since $\rho_1,\ldots,\rho_n$ are i.i.d., it follows that
$\rho^D_1,\ldots,\rho^D_n$ are i.i.d. as well. We have that, for
each $j=1,\ldots,n$
\begin{eqnarray*}
\mathbb{E}_{\rho_1,\ldots,\rho_n\in\cE_{d,d}}\Br{\rS\Pa{\frac{\Upsilon^D-\rho^D_j}{n-1}}}=\cdots=\mathbb{E}_{\rho_1,\ldots,\rho_{n-1}\in\cE_{d,d}}\Br{\rS\Pa{\frac{\sum^{n-1}_{j=1}\rho^D_j}{n-1}}}.
\end{eqnarray*}
Therefore, we get the desired inequality \eqref{eq:diag-part}. Now
we see from Proposition~\ref{prop:2-level} and
Proposition~\ref{prop:3-level} that
$\overline{\rS}_3>\overline{\rS}_2$ and
$\overline{\rS}^D_3>\overline{\rS}^D_2$. Moreover, we have that
$\overline{\rS}^D_3 - \overline{\rS}_3 <\overline{\rS}^D_2 -
\overline{\rS}_2$, i.e.,
$\overline{\sC}^{(3)}_r<\overline{\sC}^{(2)}_r$, as mentioned in
\cite{Zhang2018}. We also confirm the first strict inequality in the
\emph{conjecture} proposed in \cite{Zhang2018}:
\begin{eqnarray}\label{eq:conj-eigen}
\mathbb{E}_{\rho_1,\rho_2\in\cE_{2,2}}\Br{\rS\Pa{\frac{\rho_1+\rho_2}{2}}}&<&\mathbb{E}_{\rho_1,\rho_2,\rho_3\in\cE_{2,2}}\Br{\rS\Pa{\frac{\rho_1+\rho_2+\rho_3}{3}}}\notag\\
&<&\cdots<
\mathbb{E}_{\rho_1,\ldots,\rho_n\in\cE_{2,2}}\Br{\rS\Pa{\frac{\rho_1+\cdots+\rho_n}{n}}}.
\end{eqnarray}
Of course, we have a similar conjecture:
\begin{eqnarray}\label{eq:conj-diag}
\mathbb{E}_{\rho_1,\rho_2\in\cE_{2,2}}\Br{\rS\Pa{\frac{\rho^D_1+\rho^D_2}{2}}}&<&\mathbb{E}_{\rho_1,\rho_2,\rho_3\in\cE_{2,2}}\Br{\rS\Pa{\frac{\rho^D_1+\rho^D_2+\rho^D_3}{3}}}\notag\\
&<&\cdots<
\mathbb{E}_{\rho_1,\ldots,\rho_n\in\cE_{2,2}}\Br{\rS\Pa{\frac{\rho^D_1+\cdots+\rho^D_n}{n}}}.
\end{eqnarray}
Furthermore, we can propose the following conjecture based on the
above two observations \eqref{eq:conj-eigen} and
\eqref{eq:conj-diag}:
\begin{eqnarray}
\overline{\sC}^{(2)}_r>\overline{\sC}^{(3)}_r>\cdots>\overline{\sC}^{(n)}_r
\end{eqnarray}
for arbitrary natural number $n>3$ and
$\lim_{n\to\infty}\overline{\sC}^{(n)}_r=0$. Thus, in the qubit
case, we find that the quantum coherence monotonously decreases
statistically as the mixing times $n$. Moreover, we believe that the
quantum coherence approaches zero when $n\to\infty$. Our work
suggests that 'mixing reduces coherence'.

\section{Concluding remarks}\label{sect:con-rem}

In this paper, we relate the (equiprobable) probabilistic mixture of
adjoint orbits of quantum states to Duistermaat-Heckman measure, and
obtain theoretically the spectral density of such mixture. As an
illustration, we compute analytically the spectral densities for
mixtures consisting of 2,3,4, and 5 random qubit states. In the
qubit case, we also demonstrate the density function of a generic
eigenvalue by drawing its corresponding graph in the coordinate
system. As one application of our results, we use them to explain
why 'mixing reduces coherence' by computing the average coherence of
such mixture in the qubit state space. It is also interesting to
consider the limiting distribution of mixing arbitrary $n$
isospectral qudit density matrices.

Besides, a special case of our problem considered in
\eqref{eq:multi-orbits} is that all $\boldsymbol{\lambda}^j$ are the
same $\boldsymbol{\lambda}$. In such case,
$\boldsymbol{\lambda}(\rho_s)\prec\boldsymbol{\lambda}$, where
$\boldsymbol{\lambda}(\rho_s)$ is the vector of eigenvalues of
$\rho_s$. Inversely, Daftuar and Patrick's result \cite[Corollary
2.7.]{Daftuar2005} tells us that if a matrix $\sigma$ can be written
as a convex combination of unitary conjugations of a fixed Hermitian
matrix $\rho$ with $N$ terms, then it can be represented
equiprobable mixture of $N$ isospectral Hermitian matrices with
defined spectrum. This is the reason why we have only considered the
equiprobable mixture of quantum states. In addition, the mixture of
$N$ copies of the same quantum state corresponds to a special unital
quantum channel. This point of view can be used to investigate some
statistical properties of a random unital quantum channel in this
subclass. They also connecting Horn's problem with state
transformation in quantum information theory, which is intimately
related to LOCC interconvertion of bipartite pure states. We will
continue to study related problems along this direction in the
future research.

\subsubsection*{Acknowledgments}

This research is supported by National Natural Science Foundation of
China under Grant no.11971140, and also by Zhejiang Provincial
Natural Science Foundation of China under Grant No. LY17A010027 and
NSFC (Nos.11701259,11801123,61771174). Partial material of the
present work is completed during a research visit to Chern Institute
of Mathematics, at Nankai University. LZ are grateful to Prof.
Jing-Ling Chen for his hospitality during his visit. LZ also would
like to thank Prof. Zhi Yin for his invitation to visit Institute
for Advanced Study in Mathematics of HIT, where the discussion with
him and, in particular, with Prof. Ke Li is very helpful in solving
some puzzle later in the present paper. Finally, LZ acknowledges Hua
Xiang for improving the quality of our manuscript.


\newpage

\appendix
\appendixpage
\addappheadtotoc

\section{Derivations of the density functions for $n=3,4,5$}\label{app:a}

\subsection{Derivation of $p^{(3)}(r)$}

As already known, $\rho(\bsr)=\frac13\sum^3_{j=1}\rho_j$ can be
rewritten as via $\rho_j=\rho(\bsr_j)$
\begin{eqnarray*}
\rho(\bsr) =
\frac23\Pa{\frac{\rho(\bsr_1)+\rho(\bsr_2)}2}+\frac13\rho(\bsr_3)=
\frac23\rho(\bsr_{12})+\frac13\rho(\bsr_3).
\end{eqnarray*}
Then we see that
\begin{eqnarray*}
p^{(3)}(r) = \iint_{R_{2/3}(r)}
p_{2/3}(r|r_{12},r_3)p^{(2)}(r_{12})p^{(1)}(r_3)\dif r_{12}\dif r_3,
\end{eqnarray*}
where
\begin{eqnarray*}
p_{2/3}(r|r_{12},r_3) = \frac{9r}{4r_{12}r_3},\quad p^{(2)}(r_{12})
=
12r^2_{12}(r^3_{12}-3r_{12}+2),\quad p^{(1)}(r_3)=3r^2_3,\\
R_{2/3}(r) = \Set{(r_{12},r_3)\in[0,1]^2:
\frac{\abs{2r_{12}-r_3}}3\leqslant r\leqslant \frac{2r_{12}+r_3}3}.
\end{eqnarray*}
Denote $\Delta_3=p_{2/3}(r|r_{12},r_3)p^{(2)}(r_{12})p^{(1)}(r_3)$.\\
(1) If $r\in\Br{\frac23,1}$, then
\begin{eqnarray*}
p^{(3)}(r) = \int^1_{\frac{3r-1}2}\dif r_{12}
\int^1_{3r-2r_{12}}\dif r_3 \Delta_3= f^{(3)}_R(r).
\end{eqnarray*}
(2) If $r\in\Br{\frac13,\frac23}$, then
\begin{eqnarray*}
p^{(3)}(r) = \int^{\frac{3r}2}_{\frac{3r-1}2}\dif r_{12}
\int^1_{3r-2r_{12}}\dif r_3\Delta_3 +\int^1_{\frac{3r}2}\dif
r_{12}\int^1_{2r_{12}-3r}\dif r_3\Delta_3= f^{(3)}_R(r).
\end{eqnarray*}
(3) If $r\in\Br{\frac16,\frac13}$, then
\begin{eqnarray*}
p^{(3)}(r) = \int^{\frac{1-3r}2}_0\dif r_{12}
\int^{3r+2r_{12}}_{3r-2r_{12}}\dif
r_3\Delta_3+\int^{\frac{3r}2}_{\frac{1-3r}2}\dif
r_{12}\int^1_{3r-2r_{12}}\dif r_3\Delta_3 +
\int^{\frac{1+3r}2}_{\frac{3r}2}\dif r_{12}\int^1_{2r_{12}-3r}\dif
r_3\Delta_3= f^{(3)}_L(r).
\end{eqnarray*}
(4) If $r\in\Br{0,\frac16}$, then
\begin{eqnarray*}
p^{(3)}(r) = \int^{\frac{3r}2}_0\dif r_{12}
\int^{3r+2r_{12}}_{3r-2r_{12}}\dif r_3\Delta_3
+\int^{\frac{1-3r}2}_{\frac{3r}2}\dif
r_{12}\int^{2r_{12}+3r}_{2r_{12}-3r}\dif r_3\Delta_3 +
\int^{\frac{1+3r}2}_{\frac{1-3r}2}\dif r_{12}\int^1_{2r_{12}-3r}\dif
r_3\Delta_3= f^{(3)}_L(r).
\end{eqnarray*}
Therefore we get the desired result. $\qed$

\subsection{Derivation of $p^{(4)}(r)$}

\begin{proof}[The first proof]
We rewrite
$\rho(\bsr)=\frac14(\rho(\bsr_1)+\rho(\bsr_2)+\rho(\bsr_3)+\rho(\bsr_4))$
as
\begin{eqnarray*}
\rho(\bsr)=\frac34\Pa{\frac{\rho(\bsr_1)+\rho(\bsr_2)+\rho(\bsr_3)}3}+\frac14\rho(\bsr_4)=\frac34\rho(\bsr_{123})+\frac14\rho(\bsr_4).
\end{eqnarray*}
Then we see that
\begin{eqnarray*}
p^{(4)}(r) =  \iint_{R_{3/4}(r)}
p_{3/4}(r|r_{123},r_4)p^{(3)}(r_{123})p^{(1)}(r_4)\dif r_{123}\dif
r_4
\end{eqnarray*}
where
\begin{eqnarray*}
p_{3/4}(r|r_{123},r_4) = \frac{8r}{3r_{123}r_4},\quad
p^{(3)}(r_{123})=\begin{cases}f^{(3)}_R(r_{123}),& r_{123}\in\Br{\frac13,1},\\
f^{(3)}_L(r_{123}),&r_{123}\in\Br{0,\frac13}.
\end{cases},\quad p^{(1)}(r_4)=3r^2_4,\\
R_{3/4}(r) = \Set{(r_{123},r_4)\in[0,1]^2:
\frac{\abs{3r_{123}-r_4}}4\leqslant r\leqslant
\frac{3r_{123}+r_4}4}.
\end{eqnarray*}
(1) If $r\in\Br{\frac34,1}$, then
\begin{eqnarray*}
p^{(4)}(r) &=& \int^1_{\frac{4r-1}3}\dif
r_{123}\int^1_{4r-3r_{123}}\dif r_4
p_{3/4}(r|r_{123},r_4)f^{(3)}_R(r_{123})p^{(1)}(r_4) = f^{(4)}_R(r).
\end{eqnarray*}
(2) If $r\in\Br{\frac12,\frac34}$, then
\begin{eqnarray*}
p^{(4)}(r) &=& \Pa{\int^{\frac{4r}3}_{\frac{4r-1}3}\dif
r_{123}\int^1_{4r-3r_{123}}\dif r_4 + \int^1_{\frac{4r}3}\dif
r_{123}\int^1_{3r_{123}-4r}\dif r_4}
p_{3/4}(r|r_{123},r_4)f^{(3)}_R(r_{123})p^{(1)}(r_4) \notag\\
&=& f^{(4)}_R(r).
\end{eqnarray*}
(3) If $r\in\Br{\frac14,\frac12}$, then
\begin{eqnarray*}
p^{(4)}(r) &=& \int^{\frac13}_{\frac{4r-1}3}\dif
r_{123}\int^1_{4r-3r_{123}}\dif r_4
p_{3/4}(r|r_{123},r_4)f^{(3)}_L(r_{123})p^{(1)}(r_4)\notag\\
&& +\int^{\frac{4r}3}_{\frac13}\dif r_{123}\int^1_{4r-3r_{123}}\dif
r_4
p_{3/4}(r|r_{123},r_4)f^{(3)}_R(r_{123})p^{(1)}(r_4) \notag\\
&& +\int^{\frac{4r+1}3}_{\frac{4r}3}\dif
r_{123}\int^1_{3r_{123}-4r}\dif r_4
p_{3/4}(r|r_{123},r_4)f^{(3)}_R(r_{123})p^{(1)}(r_4)\notag\\
&=& f^{(4)}_L(r).
\end{eqnarray*}
(4) If $r\in\Br{\frac18,\frac14}$, then
\begin{eqnarray*}
p^{(4)}(r) &=& \int^{\frac{1-4r}3}_0\dif
r_{123}\int^{4r+3r_{123}}_{4r-3r_{123}}\dif r_4
p_{3/4}(r|r_{123},r_4)f^{(3)}_L(r_{123})p^{(1)}(r_4)\notag\\
&& +\int^{\frac{4r}3}_{\frac{1-4r}3}\dif
r_{123}\int^1_{4r-3r_{123}}\dif r_4
p_{3/4}(r|r_{123},r_4)f^{(3)}_L(r_{123})p^{(1)}(r_4) \notag\\
&& +\int^{\frac13}_{\frac{4r}3}\dif r_{123}\int^1_{3r_{123}-4r}\dif
r_4
p_{3/4}(r|r_{123},r_4)f^{(3)}_L(r_{123})p^{(1)}(r_4)\notag\\
&& +\int^{\frac{4r+1}3}_{\frac13}\dif
r_{123}\int^1_{3r_{123}-4r}\dif r_4
p_{3/4}(r|r_{123},r_4)f^{(3)}_R(r_{123})p^{(1)}(r_4)\notag\\
&=& f^{(4)}_L(r).
\end{eqnarray*}
(5) If $r\in\Br{0,\frac18}$, then
\begin{eqnarray*}
p^{(4)}(r) &=& \int^{\frac{4r}3}_0\dif
r_{123}\int^{4r+3r_{123}}_{4r-3r_{123}}\dif r_4
p_{3/4}(r|r_{123},r_4)f^{(3)}_L(r_{123})p^{(1)}(r_4)\notag\\
&& +\int^{\frac{1-4r}3}_{\frac{4r}3}\dif
r_{123}\int^{4r+3r_{123}}_{3r_{123}-4r}\dif r_4
p_{3/4}(r|r_{123},r_4)f^{(3)}_L(r_{123})p^{(1)}(r_4) \notag\\
&& +\int^{\frac13}_{\frac{1-4r}3}\dif
r_{123}\int^1_{3r_{123}-4r}\dif r_4
p_{3/4}(r|r_{123},r_4)f^{(3)}_L(r_{123})p^{(1)}(r_4)\notag\\
&& +\int^{\frac{4r+1}3}_{\frac13}\dif
r_{123}\int^1_{3r_{123}-4r}\dif r_4
p_{3/4}(r|r_{123},r_4)f^{(3)}_R(r_{123})p^{(1)}(r_4)\notag\\
&=& f^{(4)}_L(r).
\end{eqnarray*}
That is,
\begin{eqnarray*}
p^{(4)}(r) = \begin{cases}f^{(4)}_R(r),& r\in\Br{\frac12,1},\\
f^{(4)}_L(r),&r\in\Br{0,\frac12}.
\end{cases}
\end{eqnarray*}
We have done it.
\end{proof}

\begin{proof}[The second proof]
As already known,
$\rho(\bsr)=\frac14(\rho(\bsr_1)+\rho(\bsr_2)+\rho(\bsr_3)+\rho(\bsr_4))$
can be rewritten as
\begin{eqnarray*}
\rho(\bsr) =
\frac12\Pa{\frac{\rho(\bsr_1)+\rho(\bsr_2)}2+\frac{\rho(\bsr_3)+\rho(\bsr_4)}2}
=\frac{\rho(\bsr_{12})+\rho(\bsr_{34})}2
\end{eqnarray*}
Then we see that
\begin{eqnarray*}
p^{(4)}(r) =  \iint_{R(r)}
p(r|r_{12},r_{34})p^{(2)}(r_{12})p^{(2)}(r_{34})\dif r_{12}\dif
r_{34}
\end{eqnarray*}
(1) If $r\in\Br{\frac12,1}$, then
\begin{eqnarray*}
p^{(4)}(r) =\int^1_{2r-1}\dif r_{12}\int^1_{2r-r_{12}}\dif r_{34}
p(r|r_{12},r_{34})p^{(2)}(r_{12})p^{(2)}(r_{34})=f^{(4)}_R(r).
\end{eqnarray*}
(2) If $r\in\Br{\frac14,\frac12}$, then
\begin{eqnarray*}
p^{(4)}(r) &=& \int^{1-2r}_0\dif
r_{12}\int^{2r+r_{12}}_{2r-r_{12}}\dif r_{34}+\int^{2r}_{1-2r}\dif
r_{12}\int^1_{2r-r_{12}}\dif r_{34}+\int^1_{2r}\dif
r_{12}\int^1_{r_{12}-2r}\dif r_{34}\notag \\
&=&f^{(4)}_L(r).
\end{eqnarray*}
(3) If $r\in\Br{0,\frac14}$, then
\begin{eqnarray*}
p^{(4)}(r) &=& \int^{2r}_0\dif
r_{12}\int^{2r+r_{12}}_{2r-r_{12}}\dif r_{34}+\int^{1-2r}_{2r}\dif
r_{12}\int^{r_{12}+2r}_{r_{12}-2r}\dif r_{34}+\int^1_{1-2r}\dif
r_{12}\int^1_{r_{12}-2r}\dif r_{34} \notag\\
&=&f^{(4)}_L(r).
\end{eqnarray*}
Note that the integrand
$p(r|r_{12},r_{34})p^{(2)}(r_{12})p^{(2)}(r_{34})$ is omitted in the
case (2) and (3), respectively. In summary, we get the desired
result.
\end{proof}

\subsection{Derivation of $p^{(5)}(r)$}

Note that
\begin{eqnarray*}
\rho(\bsr) =
\frac45\Pa{\frac14\sum^4_{j=1}\rho(\bsr_j)}+\frac15\rho(\bsr_5) =
\frac45\sigma_1+\frac15\sigma_2,
\end{eqnarray*}
where $\sigma_1=\frac14\sum^4_{j=1}\rho(\bsr_j)$ and
$\sigma_2=\rho(\bsr_5)$. From this, we see that
\begin{eqnarray*}
p^{(5)}(r) =  \iint_{R_{4/5}(r)}
p_{4/5}(r|r_1,r_2)p^{(4)}(r_1)p^{(1)}(r_2)\dif r_1\dif r_2.
\end{eqnarray*}
(1) If $r\in\Br{\frac45,1}$, then
\begin{eqnarray*}
\int^1_{\frac{5r - 1}4}\dif r_1\int^1_{5r - 4r_1} \dif r_2
p_{4/5}(r|r_1,r_2)f^{(4)}_R(r_1)p^{(1)}(r_2)= f^{(5)}_R(r).
\end{eqnarray*}
(2) If $r\in\Br{\frac35,\frac45}$, then
\begin{eqnarray*}
\Pa{\int^{\frac{5r}4}_{\frac{5r - 1}4}\dif r_1\int^1_{5r - 4r_1}
\dif r_2 + \int^1_{\frac{5r}4}\dif r_1\int^1_{4r_1 - 5r} \dif r_2}
p_{4/5}(r|r_1,r_2)f^{(4)}_R(r_1)p^{(1)}(r_2)= f^{(5)}_R(r).
\end{eqnarray*}
(3) If $r\in\Br{\frac25,\frac35}$, then
\begin{eqnarray*}
&&\int^{\frac12}_{\frac{5r - 1}4}\dif r_1\int^1_{5r - 4r_1} \dif
r_2p_{4/5}(r|r_1,r_2)f^{(4)}_L(r_1)p^{(1)}(r_2) \\
&&+ \int^{\frac{5r}4}_{\frac12}\dif r_1\int^1_{5r-4r_1} \dif r_2
p_{4/5}(r|r_1,r_2)f^{(4)}_R(r_1)p^{(1)}(r_2)\\
&&+ \int^{\frac{5r+1}4}_{\frac{5r}4}\dif r_1\int^1_{4r_1-5r} \dif
r_2 p_{4/5}(r|r_1,r_2)f^{(4)}_R(r_1)p^{(1)}(r_2)= f^{(5)}_M(r).
\end{eqnarray*}
(4) If $r\in\Br{\frac15,\frac25}$, then
\begin{eqnarray*}
&&\int^{\frac{5r}4}_{\frac{5r - 1}4}\dif r_1\int^1_{5r - 4r_1} \dif
r_2p_{4/5}(r|r_1,r_2)f^{(4)}_L(r_1)p^{(1)}(r_2) \\
&&+ \int^{\frac12}_{\frac{5r}4}\dif r_1\int^1_{4r_1-5r} \dif r_2
p_{4/5}(r|r_1,r_2)f^{(4)}_L(r_1)p^{(1)}(r_2)\\
&&+ \int^{\frac{5r+1}4}_{\frac12}\dif r_1\int^1_{4r_1-5r} \dif r_2
p_{4/5}(r|r_1,r_2)f^{(4)}_R(r_1)p^{(1)}(r_2)= f^{(5)}_M(r).
\end{eqnarray*}
(5) If $r\in\Br{\frac1{10},\frac15}$, then
\begin{eqnarray*}
&&\int^{\frac{1-5r}4}_0\dif r_1\int^{5r + 4r_1}_{5r - 4r_1} \dif
r_2p_{4/5}(r|r_1,r_2)f^{(4)}_L(r_1)p^{(1)}(r_2) \\
&&+ \int^{\frac{5r}4}_{\frac{1-5r}4}\dif r_1\int^1_{5r-4r_1} \dif
r_2
p_{4/5}(r|r_1,r_2)f^{(4)}_L(r_1)p^{(1)}(r_2)\\
&&+ \int^{\frac{5r+1}4}_{\frac{5r}4}\dif r_1\int^1_{4r_1-5r} \dif
r_2 p_{4/5}(r|r_1,r_2)f^{(4)}_L(r_1)p^{(1)}(r_2)= f^{(5)}_L(r).
\end{eqnarray*}
(6) If $r\in\Br{0,\frac1{10}}$, then
\begin{eqnarray*}
&&\int^{\frac{5r}4}_0\dif r_1\int^{5r + 4r_1}_{5r - 4r_1} \dif
r_2p_{4/5}(r|r_1,r_2)f^{(4)}_L(r_1)p^{(1)}(r_2) \\
&&+ \int^{\frac{1-5r}4}_{\frac{5r}4}\dif r_1\int^{5r +
4r_1}_{4r_1-5r} \dif r_2
p_{4/5}(r|r_1,r_2)f^{(4)}_L(r_1)p^{(1)}(r_2)\\
&&+ \int^{\frac{5r+1}4}_{\frac{1-5r}4}\dif r_1\int^1_{4r_1-5r} \dif
r_2 p_{4/5}(r|r_1,r_2)f^{(4)}_L(r_1)p^{(1)}(r_2)= f^{(5)}_L(r).
\end{eqnarray*}
Thus we get the result. $\qed$ \\
Note that in the above reasoning, the symbolic computation function
of the computer software \textsc{Mathematica 10} are employed in
almost all calculations.

\end{document}